\documentclass[a4paper,UKenglish,cleveref, autoref, thm-restate]{lipics-v2021}

\author{Timo Gervens}{RWTH Aachen University, 
	Germany}{gervens@informatik.rwth-aachen.de}{https://orcid.org/0000-0002-1224-9853}{European Union (ERC, SymSim, 101054974)}

\author{Martin Grohe}{RWTH Aachen University, 
	Germany}{grohe@informatik.rwth-aachen.de}{https://orcid.org/0000-0002-0292-9142}{European Union (ERC, SymSim, 101054974)}

\author{Louis Härtel}{RWTH Aachen University, 
	Germany}{haertel@informatik.rwth-aachen.de}{https://orcid.org/0009-0004-3446-5874}{European Union (ERC, SymSim, 101054974)}
    
\author{Philipp da Silva Fonseca}{RWTH Aachen University, 
	Germany}{}{https://orcid.org/0009-0006-4883-7274}{}

\authorrunning{T. Gervens, M. Grohe, L. Härtel, and P. da Silva Fonseca} 

\Copyright{Timo Gervens, Martin Grohe, Louis Härtel, Philipp da Silva Fonseca} 

\title{The Complexity of Homomorphism Reconstruction Revisited}

\keywords{graph homomorphism, nexp-complete, counting complexity}

\ccsdesc[500]{Theory of computation~Graph algorithms analysis}
\ccsdesc[300]{Theory of computation~Parameterized complexity and exact 
algorithms}
\ccsdesc[500]{Mathematics of computing~Graph theory}

\relatedversion{}

\nolinenumbers

\hideLIPIcs

\usepackage{mathtools}
\usepackage{todonotes}    \usepackage{multicol}
\usepackage{csquotes}
\usepackage{amsmath}
\usepackage[ruled]{algorithm2e}
\usepackage{tikz}
\usetikzlibrary{calc,
	cd,
	decorations.pathreplacing,
	decorations.text,
	calligraphy,
	positioning}
\tikzstyle{vertex}=[circle, draw=gray, fill=gray,inner sep=0pt, minimum size=2mm]
\tikzstyle{smallvertex}=[circle, draw=gray, fill=gray,inner sep=0pt, minimum size=1.2mm]
\tikzstyle{kneser}=[circle, draw=gray, minimum size=0.8cm]
\tikzstyle{cycle}=[circle, draw=gray, minimum size=1.5cm]

\newsavebox{\fminibox}
\newlength{\fminilength}
\newenvironment{fminipage}[1][\linewidth]{
	\setlength{\fminilength}{#1-2\fboxsep-2\fboxrule}\begin{lrbox}{\fminibox}\begin{minipage}{\fminilength}}{
	\end{minipage}\end{lrbox}\noindent\fbox{\usebox{\fminibox}}
}
\usepackage{xspace}
\usepackage{stmaryrd}

\newcommand{\abs}[1]{\left\lvert #1 \right\rvert}

\newcommand{\val}[1]{\left\llbracket #1 \right\rrbracket}
\renewcommand\phi\varphi
\renewcommand\epsilon\varepsilon

\DeclareMathOperator{\sub}{sub}

\newcommand{\cons}[1]{\mathcal{#1}}
\newcommand{\bin}[1]{\overline{#1}}

\newcommand{\CC}{\mathcal{C}}

\newcommand{\CF}{\mathcal{F}}

\newcommand{\CI}{\mathcal{I}}
\newcommand{\CL}{\mathcal{L}}

\newcommand{\NN}{\mathbb{N}}

\newcommand{\coNP}{\mathsf{coNP}}

\newcommand{\NEXP}{\textsf{\upshape NEXP}}
\newcommand{\NP}{\textsf{NP}}

\newcommand{\SP}{\#\mathsf{P}}

\newcommand{\N}{\ensuremath{\mathbb{N}}}

\DeclareMathOperator{\h}{h}
\DeclareMathOperator{\s}{s}
\DeclareMathOperator{\dpf}{DP}

\newcommand{\Fnm}{\CF_\equiv}
\newcommand{\cupp}{\; \cup \;}

\newcommand{\Cir}{\mathsf{C}}
\newcommand{\Bit}{\mathsf{B}}
\newcommand{\Inp}{\mathsf{In}}

\newcommand{\PPP}{\mathsf{P}}
\newcommand{\QQQ}{\mathsf{Q}}
\newcommand{\XXX}{\mathsf{X}}
\newcommand{\SSS}{\mathsf{S}}
\newcommand{\TTT}{\mathsf{T}}
\newcommand{\AAA}{\mathsf{A}}
\newcommand{\BBB}{\mathsf{B}}

    \newcommand{\dproblem}[3]{
	\begin{center}
		\begin{fminipage}[.95\linewidth]
			\textup{\textsc{#1}\begin{description}
				\item[Input:] #2
				\item[Question:] #3
			\end{description}}
\end{fminipage}
\end{center}}
\newcommand{\pproblem}[4]{
	\begin{center}
		\begin{fminipage}[.95\linewidth]
			\textup{\textsc{#1}\begin{description}
				\item[Input:] #2
				\item[Parameter:] #3
				\item[Question:] #4
			\end{description}}
\end{fminipage}
\end{center}}

\newcommand{\HomR}{\textsc{HomRec}}
\newcommand{\UnHomR}{\textsc{UnHomRec}}

\newcommand{\DecProblemA}{\textup{\textsc{A}}\xspace}
\newcommand{\DecProblemB}{\textup{\textsc{B}}\xspace}

\newcommand{\nonthreecolouring}{\textup{\textsc{Non-$3$-Colouring}}\xspace}

\newcommand{\NonThreeColouringExtension}{\textup{\textsc{$2$-Round-$3$-Colouring}}\xspace}
\newcommand{\CircuitSAT}{\textup{\textsc{CircuitSat}}\xspace}
\newcommand{\SuccinctClique}{\textup{\textsc{SuccinctClique}}\xspace}

\renewcommand\theta\vartheta
\renewcommand\phi\varphi
\renewcommand\epsilon\varepsilon

\tikzset{
  smallvertex/.style={circle,draw,fill,inner sep=1.2pt},
  snode/.style={black,draw,circle,fill,inner sep = 0mm,minimum
    height=1mm}
  }

\newcommand{\smallKone}{\begin{tikzpicture}[baseline=-3pt]
		\node[smallvertex] (R) {};\end{tikzpicture}}
\newcommand{\smallKtwo}{{\begin{tikzpicture}[baseline=-3pt]
		\node[smallvertex] (R) {};\node[smallvertex, left = 0.2cm of R] (M) {};\draw (M) edge (R);
\end{tikzpicture}}}

\newcommand{\smallKthree}{
\begin{tikzpicture}[baseline=.5pt]
		\node[smallvertex] (R) {};\node[smallvertex, left = 0.1cm of R] (M) {};\node[smallvertex, fill = white, draw = white, above = 0.1cm of R] (Tghost) {};\node[smallvertex, fill = white, draw = white, above = 0.1cm of M] (Sghost) {};\node[smallvertex] (T) at ($(Tghost)!0.5!(Sghost)$){};
		
		\draw (M) edge (R);
		\draw (T) edge (R);
		\draw (M) edge (T);
\end{tikzpicture}}

\newcommand{\smallPtwo}{\begin{tikzpicture}[baseline=.5pt]
		\node[smallvertex] (R) {};\node[smallvertex, left = 0.1cm of R] (M) {};\node[smallvertex, fill = white, draw = white, above = 0.1cm of R] (Tghost) {};\node[smallvertex, fill = white, draw = white, above = 0.1cm of M] (Sghost) {};\node[smallvertex] (T) at ($(Tghost)!0.5!(Sghost)$){};
		
		\draw (M) edge (R);
		\draw (M) edge (T);
\end{tikzpicture}}

\newcommand{\smallAlpha}{\begin{tikzpicture}[baseline=.5pt]
		\node[smallvertex, label={above:\small$\bot$}] (A) {};\node[smallvertex, label={above:\small$\alpha$}, right = 0.2cm of A] (B) {};\node[smallvertex, label={above:\small$\alpha$}, right = 0.2cm of B] (C) {};\node[smallvertex, label={above:\small$\top$}, right = 0.2cm of C] (D) {};

		\draw (A) edge (B);
		\draw (B) edge (C);
        \draw (C) edge (D);
\end{tikzpicture}}
\newcommand{\smallColour}[1]{\begin{tikzpicture}[baseline=.5pt]
  \node[smallvertex, label={above:\small$\SSS$}] (A) {};\node[smallvertex, label={above:\small$\XXX$}, right = 0.2cm of A] (B) {};\node (dots) [right=0.2cm of B] {$\cdots$}; \node[smallvertex, label={above:\small$\XXX$}, right = 0.2cm of dots] (D) {};\node[smallvertex, label={above:\small$\TTT$}, right = 0.2cm of D] (E) {};

\draw (A) -- (B);
  \draw (B) -- (dots);
  \draw (dots) -- (D);
  \draw (D) -- (E);

\draw[decorate,decoration={brace,amplitude=4pt}]
    (B.north) -- (D.north) node[midway,above=6pt] {$#1$};
\end{tikzpicture}}
\newcommand{\smallColourCycle}[3]{\begin{tikzpicture}[baseline=.5pt]
  \node[smallvertex, label={above:\small$\SSS$}] (S) {};\node[smallvertex, label={above:\small$\XXX$}, right = 0.2cm of S] (B) {};\node (dots) [right=0.2cm of B] {$\cdots$}; \node[smallvertex, label={above:\small$\XXX$}, right = 0.2cm of dots] (D) {};\node[smallvertex, label={above:\small$\TTT$}, right = 0.2cm of D] (E) {};

  \node[smallvertex, label={right:\small$\AAA$}, below = 0.2cm of dots] (A) {};

\node[right = 0.1cm of B] (I) {};\node[left = 0.1cm of D] (J) {};

\draw (S) -- (B);
  \draw (B) -- (dots);
  \draw (dots) -- (D);
  \draw (D) -- (E);
  \draw (A) -- (I);
  \draw (A) -- (J);
    
\draw[decorate,decoration={brace,amplitude=4pt}]
    (B.north) -- (D.north) node[midway,above=6pt] {$#1$};
\end{tikzpicture}}
\newcommand{\smallAlphaCfour}[1]{\begin{tikzpicture}[baseline=.5pt]
\node[smallvertex, label={left:\small$\Cir_v$}] (A) {};
    \node[smallvertex, right=0.2cm of A, label={right:\small$\Bit_{#1}$}] (B) {};
    \node[smallvertex, above=0.2cm of B, label={right:\small$\alpha$}] (C) {};
    \node[smallvertex, above=0.2cm of A, label={left:\small$\alpha$}] (D) {};
    
\draw (A) -- (B);
    \draw (B) -- (C);
    \draw (C) -- (D);
    \draw (D) -- (A);
\end{tikzpicture}}
\newcommand{\smallAlphaAnd}[3]{\begin{tikzpicture}[baseline=.5pt]
		\node[smallvertex, label={above:\small$\bot$}] (A) {};\node[smallvertex, label={above:\small$\alpha$}, right = 0.2cm of A] (0) {};\node[smallvertex, label={above:\small$\alpha$}, right = 0.2cm of 0] (1) {};\node[smallvertex, label={above:\small$\top$}, right = 0.2cm of 1] (D) {};

		\draw (A) edge (0);
		\draw (0) edge (1);
        \draw (1) edge (D);

        \node[smallvertex, below left = 0.15cm of 0, label={below:\small$\Cir_u$}] (R) {};\node[smallvertex, right = 0.3cm of R, label={below:\small$\Cir_w$}] (M) {};\node[smallvertex, right = 0.3cm of M, label={below:\small$\Cir_v$}] (T) {};

\draw (R) -- (#1);
        \draw (M) -- (#2);
        \draw (T) -- (#3);
\end{tikzpicture}}
\newcommand{\smallAlphaOr}[3]{\begin{tikzpicture}[baseline=.5pt]
		\node[smallvertex, label={above:\small$\bot$}] (A) {};\node[smallvertex, label={above:\small$\alpha$}, right = 0.2cm of A] (0) {};\node[smallvertex, label={above:\small$\alpha$}, right = 0.2cm of 0] (1) {};\node[smallvertex, label={above:\small$\top$}, right = 0.2cm of 1] (D) {};

		\draw (A) edge (0);
		\draw (0) edge (1);
        \draw (1) edge (D);

        \node[smallvertex, below left = 0.15cm of 0, label={below:\small$\Cir_u$}] (R) {};\node[smallvertex, right = 0.3cm of R, label={below:\small$\Cir_w$}] (M) {};\node[smallvertex, right = 0.3cm of M, label={below:\small$\Cir_v$}] (T) {};

\draw (R) -- (#1);
        \draw (M) -- (#2);
        \draw (T) -- (#3);
\end{tikzpicture}}
\newcommand{\smallAlphaNeg}[2]{\begin{tikzpicture}[baseline=.5pt]
		\node[smallvertex, label={above:\small$\bot$}] (A) {};\node[smallvertex, label={above:\small$\alpha$}, right = 0.2cm of A] (0) {};\node[smallvertex, label={above:\small$\alpha$}, right = 0.2cm of 0] (1) {};\node[smallvertex, label={above:\small$\top$}, right = 0.2cm of 1] (D) {};

		\draw (A) edge (0);
		\draw (0) edge (1);
        \draw (1) edge (D);

        \node[smallvertex, below left = 0.15cm of 0, label={below:\small$\Cir_w$}] (R) {};\node[smallvertex, right = 0.6cm of R, label={below:\small$\Cir_v$}] (M) {};

\draw (R) -- (#1);
        \draw (M) -- (#2);
\end{tikzpicture}}

\newcommand{\smallAlphaAndX}[3]{\begin{tikzpicture}[baseline=.5pt]
		\node[smallvertex, label={above:\small$\bot$}] (A) {};\node[smallvertex, label={above:\small$\alpha$}, right = 0.2cm of A] (0) {};\node[smallvertex, label={above:\small$\alpha$}, right = 0.2cm of 0] (1) {};\node[smallvertex, label={above:\small$\top$}, right = 0.2cm of 1] (D) {};

		\draw (A) -- (0);
		\draw (0) -- (1);
        \draw (1) -- (D);

        \node[smallvertex, below left = 0.15cm of 0, label={below:\small$\Cir_u$}] (R) {};\node[smallvertex, right = 0.3cm of R, label={below:\small$\Cir_w$}] (M) {};\node[smallvertex, right = 0.3cm of M, label={below:\small$\Cir_v$}] (T) {};

\node[right = 0.2cm of R] (G) {};

\draw (R) -- (#1);
        \draw (M) -- (#2);
        \draw (T) -- (#3);

\node[smallvertex, above=0.6cm of G, label={above:\small$\PPP$}] (CG) {};\draw (CG) -- (R);
        \draw (CG) -- (M);
        \draw (CG) -- (T);
\end{tikzpicture}}

\newcommand{\smallAlphaOrX}[3]{\begin{tikzpicture}[baseline=.5pt]
		\node[smallvertex, label={above:\small$\bot$}] (A) {};\node[smallvertex, label={above:\small$\alpha$}, right = 0.2cm of A] (0) {};\node[smallvertex, label={above:\small$\alpha$}, right = 0.2cm of 0] (1) {};\node[smallvertex, label={above:\small$\top$}, right = 0.2cm of 1] (D) {};

		\draw (A) -- (0);
		\draw (0) -- (1);
        \draw (1) -- (D);

        \node[smallvertex, below left = 0.15cm of 0, label={below:\small$\Cir_u$}] (R) {};\node[smallvertex, right = 0.3cm of R, label={below:\small$\Cir_w$}] (M) {};\node[smallvertex, right = 0.3cm of M, label={below:\small$\Cir_v$}] (T) {};

\node[right = 0.2cm of R] (G) {};

\draw (R) -- (#1);
        \draw (M) -- (#2);
        \draw (T) -- (#3);

\node[smallvertex, above=0.6cm of G, label={above:\small$\PPP$}] (CG) {};\draw (CG) -- (R);
        \draw (CG) -- (M);
        \draw (CG) -- (T);
\end{tikzpicture}}

\newcommand{\smallAlphaNegX}[2]{\begin{tikzpicture}[baseline=.5pt]
		\node[smallvertex, label={above:\small$\bot$}] (A) {};\node[smallvertex, label={above:\small$\alpha$}, right = 0.2cm of A] (0) {};\node[smallvertex, label={above:\small$\alpha$}, right = 0.2cm of 0] (1) {};\node[smallvertex, label={above:\small$\top$}, right = 0.2cm of 1] (D) {};

		\draw (A) -- (0);
		\draw (0) -- (1);
        \draw (1) -- (D);

        \node[smallvertex, below left = 0.15cm of 0, label={below:\small$\Cir_w$}] (R) {};\node[smallvertex, right = 0.6cm of R, label={below:\small$\Cir_v$}] (M) {};

\node[right = 0.2cm of R] (G) {};

\draw (R) -- (#1);
        \draw (M) -- (#2);

\node[smallvertex, above=0.6cm of G, label={above:\small$\PPP$}] (CG) {};\draw (CG) -- (R);
        \draw (CG) -- (M);
\end{tikzpicture}}

\newcommand{\smallPtwoOut}{\begin{tikzpicture}[baseline=.5pt]
		\node[smallvertex, label={above:\small$\Cir_o$}] (R) {};\node[smallvertex, right = 0.2cm of R, label={above:\small$\alpha$}] (M) {};\node[smallvertex, right = 0.2cm of M, label={above:\small$\top$}] (T) {};

		\draw (M) edge (R);
		\draw (M) edge (T);
\end{tikzpicture}}

\newcommand{\smallPtwoLabelled}[3]{\begin{tikzpicture}[baseline=.5pt]
    \node[smallvertex, label={above:\small$#1$}] (R) {};\node[smallvertex, right = 0.2cm of R, label={above:\small$#2$}] (M) {};\node[smallvertex, right = 0.2cm of M, label={above:\small$#3$}] (T) {};

    \draw (M) edge (R);
    \draw (M) edge (T);
\end{tikzpicture}}
\newcommand{\diamondGadgetA}{\begingroup
  \begin{tikzpicture}[baseline=0pt,
    smallvertex/.style={circle,draw,fill,inner sep=1.2pt},
    every label/.style={font=\footnotesize}
  ]
\node[smallvertex,label=above:{\small $\PPP$}] (Q) at (0,1.2) {};
    \node[smallvertex,label=below:{\small $\QQQ$}] (P) at (0,-1.2) {};

\node[smallvertex,label=left:{\small$\Cir_{v_0}$}] (a1) at (-1.6,0.45) {};
    \node[smallvertex,label=left:{\small $\Bit_{0}$}] (b1) at (-1.6,-0.45) {};
    \draw (Q) -- (a1) -- (b1) -- (P);

\node[smallvertex,label=left:{\small$\Cir_{v_1}$}] (a2) at (0,0.45) {};
    \node[smallvertex,label=left:{\small$\Bit_{1}$}] (b2) at (0,-0.45) {};
    \draw (Q) -- (a2) -- (b2) -- (P);

\node at (0.9,0.45) {\small $\dots$};
    \node at (0.9,-0.45) {\small $\dots$};

\node[smallvertex,label=right:{\small$\Cir_{v_{n-1}}$}] (a3) at (1.8,0.45) {};
    \node[smallvertex,label=right:{\small$\Bit_{n-1}$}] (b3) at (1.8,-0.45) {};
    \draw (Q) -- (a3) -- (b3) -- (P);
  \end{tikzpicture}
  \endgroup
}
\newcommand{\diamondGadgetB}{\begingroup
  \begin{tikzpicture}[baseline=0pt,
    smallvertex/.style={circle,draw,fill,inner sep=1.2pt},
    every label/.style={font=\footnotesize}
  ]
\node[smallvertex,label=above:{\small $\PPP$}] (Q) at (0,1.2) {};
    \node[smallvertex,label=below:{\small $\QQQ$}] (P) at (0,-1.2) {};

\node[smallvertex,label=left:{\small$\Cir_{v_n}$}] (a1) at (-1.6,0.45) {};
    \node[smallvertex,label=left:{\small$\Bit_{0}$}] (b1) at (-1.6,-0.45) {};
    \draw (Q) -- (a1) -- (b1) -- (P);

\node[smallvertex,label=left:{\small$\Cir_{v_{n+1}}$}] (a2) at (0,0.45) {};
    \node[smallvertex,label=left:{\small$\Bit_{1}$}] (b2) at (0,-0.45) {};
    \draw (Q) -- (a2) -- (b2) -- (P);

\node at (0.9,0.45) {\small $\dots$};
    \node at (0.9,-0.45) {\small $\dots$};

\node[smallvertex,label=right:{\small$\Cir_{v_{2n-1}}$}] (a3) at (1.8,0.45) {};
    \node[smallvertex,label=right:{\small$\Bit_{n-1}$}] (b3) at (1.8,-0.45) {};
    \draw (Q) -- (a3) -- (b3) -- (P);
  \end{tikzpicture}
  \endgroup
}
\newcommand{\nexpCircuitGraph}{\begingroup
  \begin{tikzpicture}[x=0.75pt,y=0.75pt,yscale=-1,xscale=1]
\draw  [dash pattern={on 4.5pt off 4.5pt}] (130.0, 40.0) -- 
	(233.0, 150.0) -- (30.0, 150.0) -- cycle ;
\node[smallvertex] at (55.0,68.0) {};
\draw    (55.0, 68.0) -- (30.0, 150.0) ;
\node[smallvertex] at (130.0,40.0) {};
\draw    (55.0, 68.0) -- (132.0, 40.0) ;
\draw    (55.0, 68.0) -- (66.0, 144.0) ;
\draw    (55.0, 68.0) -- (103.0, 124.0) ;
\draw    (55.0, 68.0) -- (129.0, 83.0) ;
\draw    (55.0, 68.0) -- (133.0, 60.0) ;
\draw    (55.0, 68.0) -- (119.0, 104.0) ;
\draw    (55.0, 68.0) -- (85.0, 135.0) ;
\draw    (55.0, 68.0) -- (49.0, 148.0) ;
\draw  [dash pattern={on 4.5pt off 4.5pt}] (30.0, 150.0) -- (233.0, 150.0) -- 
	(233.0, 184.0) -- (30.0, 184.0) -- cycle ;
\draw  [dash pattern={on 4.5pt off 4.5pt}]  (131.0, 150.0) -- 
	(131.0, 160.0) -- (131.0, 184.0) ;
\node[smallvertex] at (80.0,265.0) {};
\node[smallvertex] at (30.0,211.0) {};
\node[smallvertex] at (128.0,211.0) {};
\node[smallvertex] at (30.0,184.0) {};
\node[smallvertex] at (128.0,184.0) {};
\node[smallvertex] at (134.0,184.0) {};
\draw    (30.0, 211.0) -- (30.0, 184.0) ;
\draw    (128.0, 184.0) -- (128.0, 211.0) ;
\draw    (30.0, 211.0) -- (80.0, 265.0) ;
\draw    (80.0, 265.0) -- (128.0, 211.0) ;
\draw  [dash pattern={on 4.5pt off 4.5pt}] (375.0, 40.0) -- 
	(476.0, 150.0) -- (273.0, 150.0) -- cycle ;
\node[smallvertex] at (298.0,68.0) {};
\draw    (298.0, 68.0) -- (273.0, 150.0) ;
\node[smallvertex] at (375.0,40.0) {};
\draw    (298.0, 68.0) -- (375.0, 40.0) ;
\draw    (298.0, 68.0) -- (309.0, 145.0) ;
\draw    (298.0, 68.0) -- (346.0, 124.0) ;
\draw    (298.0, 68.0) -- (372.0, 83.0) ;
\draw    (298.0, 68.0) -- (376.0, 60.0) ;
\draw    (298.0, 68.0) -- (362.0, 104.0) ;
\draw    (298.0, 68.0) -- (328.0, 136.0) ;
\draw    (298.0, 68.0) -- (292.0, 149.0) ;
\draw  [dash pattern={on 4.5pt off 4.5pt}] (273.0, 150.0) -- 
	(476.0, 150.0) -- (476.0, 184.0) -- (273.0, 184.0) -- cycle ;
\draw  [dash pattern={on 4.5pt off 4.5pt}]  (374.0, 150.0) -- 
	(374.0, 160.0) -- (374.0, 184.0) ;
\node[smallvertex] at (427.0,265.0) {};
\node[smallvertex] at (377.0,211.0) {};
\node[smallvertex] at (475.0,211.0) {};
\node[smallvertex] at (273.0,184.0) {};
\node[smallvertex] at (371.0,184.0) {};
\node[smallvertex] at (377.0,184.0) {};
\draw    (377.0, 211.0) -- (377.0, 184.0) ;
\draw    (475.0, 184.0) -- (475.0, 211.0) ;
\draw    (377.0, 211.0) -- (427.0, 265.0) ;
\draw    (427.0, 265.0) -- (475.0, 211.0) ;
\node[smallvertex] at (233.0,184.0) {};
\node[smallvertex] at (475.0,184.0) {};
\node[smallvertex] at (254.0,265.0) {};
\node[smallvertex] at (204.0,211.0) {};
\node[smallvertex] at (302.0,211.0) {};
\draw    (204.0, 211.0) -- (134.0, 184.0) ;
\draw    (233.0, 184.0) -- (302.0, 211.0) ;
\draw    (204.0, 211.0) -- (254.0, 265.0) ;
\draw    (254.0, 265.0) -- (302.0, 211.0) ;
\draw    (273.0, 184.0) -- (204.0, 211.0) ;
\draw    (371.0, 184.0) -- (302.0, 211.0) ;
\draw  [dash pattern={on 4.5pt off 4.5pt}] (175.0, 14.0) -- (265.0, 14.0) 
	-- (265.0, 54.0) -- (175.0, 54.0) -- cycle ;
\node[smallvertex] at (190.0,40.0) {};
\node[smallvertex] at (210.0,40.0) {};
\node[smallvertex] at (230.0,40.0) {};
\node[smallvertex] at (250.0,40.0) {};
\draw    (190.0, 40.0) -- (210.0, 40.0) ;
\draw    (210.0, 40.0) -- (230.0, 40.0) ;
\draw    (230.0, 40.0) -- (250.0, 40.0) ;
\draw    (210.0, 40.0) -- (176.0, 65.0) ;
\draw    (210.0, 40.0) -- (182.0, 71.0) ;
\draw    (210.0, 40.0) -- (190.0, 75.0) ;
\draw    (210.0, 40.0) -- (200.0, 77.0) ;
\draw    (210.0, 40.0) -- (210.0, 78.0) ;
\draw    (210.0, 40.0) -- (220.0, 77.0) ;
\draw    (210.0, 40.0) -- (230.0, 74.0) ;
\draw    (210.0, 40.0) -- (239.0, 70.0) ;
\draw    (210.0, 40.0) -- (244.0, 64.0) ;
\draw    (230.0, 40.0) -- (196.0, 65.0) ;
\draw    (230.0, 40.0) -- (202.0, 71.0) ;
\draw    (230.0, 40.0) -- (210.0, 75.0) ;
\draw    (230.0, 40.0) -- (220.0, 77.0) ;
\draw    (230.0, 40.0) -- (230.0, 78.0) ;
\draw    (230.0, 40.0) -- (240.0, 77.0) ;
\draw    (230.0, 40.0) -- (250.0, 74.0) ;
\draw    (230.0, 40.0) -- (259.0, 70.0) ;
\draw    (230.0, 40.0) -- (264.0, 64.0) ;
\draw    (80.0, 265.0) -- (48.0, 211.0) ;
\draw    (80.0, 265.0) -- (110.0, 212.0) ;
\draw    (48.0, 211.0) -- (49.0, 184.0) ;
\draw    (110.0, 212.0) -- (110.0, 185.0) ;
\draw    (427.0, 265.0) -- (395.0, 211.0) ;
\draw    (427.0, 265.0) -- (457.0, 212.0) ;
\draw    (395.0, 211.0) -- (396.0, 184.0) ;
\draw    (457.0, 212.0) -- (457.0, 185.0) ;
\draw    (254.0, 265.0) -- (222.0, 211.0) ;
\draw    (254.0, 265.0) -- (284.0, 212.0) ;
\draw    (222.0, 211.0) -- (153.0, 184.0) ;
\draw    (284.0, 212.0) -- (214.0, 185.0) ;
\draw    (291.0, 184.0) -- (222.0, 211.0) ;
\draw    (353.0, 185.0) -- (284.0, 212.0) ;
	
\draw (121.0, 108.0) node [anchor=north west][inner sep=0.75pt]  
	[font=\Large] [align=left] {$\displaystyle C_{G}$};
\draw (38.0, 53.0) node [anchor=north west][inner sep=0.75pt]  
	[font=\normalsize] [align=left] {$\PPP$};
\draw (121.0, 19.0) node [anchor=north west][inner sep=0.75pt]  
	[font=\Large] [align=left] {$\Cir_o$};
\draw (51.0, 254.0) node [anchor=north west][inner sep=0.75pt]  
	[font=\normalsize] [align=left] {$\QQQ$};
\draw (9.0, 214.0) node [anchor=north west][inner sep=0.75pt]   
	[align=left] {$\Bit\displaystyle _{0}$};
\draw (130.0, 214.0) node [anchor=north west][inner sep=0.75pt]   
	[align=left] {$\Bit\displaystyle _{n-1}$};
\draw (32.0, 162.0) node [anchor=north west][inner sep=0.75pt]   
	[align=left] {$\Cir\displaystyle _{v_0}$};
\draw (169.0, 156.0) node [anchor=north west][inner sep=0.75pt]  
	[font=\Large] [align=left] {$\displaystyle \overline{j}$};
\draw (135.0, 162.0) node [anchor=north west][inner sep=0.75pt]   
	[align=left] {$\Cir\displaystyle _{v_n}$};
\draw (192.0, 162.0) node [anchor=north west][inner sep=0.75pt]   
	[align=left] {$\Cir\displaystyle _{v_{2n-1}}$};
\draw (364.0, 109.0) node [anchor=north west][inner sep=0.75pt]  
	[font=\Large] [align=left] {$\displaystyle C_{G}$};
\draw (281.0, 53.0) node [anchor=north west][inner sep=0.75pt]  
	[font=\normalsize] [align=left] {$\PPP$};
\draw (366.0, 19.0) node [anchor=north west][inner sep=0.75pt]  
	[font=\Large] [align=left] {$\Cir_o$};
\draw (309.0, 155.0) node [anchor=north west][inner sep=0.75pt]  
	[font=\Large] [align=left] {$\displaystyle \overline{i}$};
\draw (398.0, 254.0) node [anchor=north west][inner sep=0.75pt]  
	[font=\normalsize] [align=left] {$\QQQ$};
\draw (356.0, 214.0) node [anchor=north west][inner sep=0.75pt]   
	[align=left] {$\Bit\displaystyle _{0}$};
\draw (477.0, 214.0) node [anchor=north west][inner sep=0.75pt]   
	[align=left] {$\Bit\displaystyle _{n-1}$};
\draw (275.0, 162.0) node [anchor=north west][inner sep=0.75pt]   
	[align=left] {$\Cir\displaystyle _{v_0}$};
\draw (332.0, 162.0) node [anchor=north west][inner sep=0.75pt]   
	[align=left] {$\Cir\displaystyle _{v_{n-1}}$};
\draw (412.0, 156.0) node [anchor=north west][inner sep=0.75pt]  
	[font=\Large] [align=left] {$\displaystyle \overline{j}$};
\draw (378.0, 162.0) node [anchor=north west][inner sep=0.75pt]   
	[align=left] {$\Cir\displaystyle _{v_{n}}$};
\draw (435.0, 162.0) node [anchor=north west][inner sep=0.75pt]   
	[align=left] {$\Cir\displaystyle _{v_{2n-1}}$};
\draw (225.0, 254.0) node [anchor=north west][inner sep=0.75pt]  
	[font=\normalsize] [align=left] {$\QQQ$};
\draw (183.0, 214.0) node [anchor=north west][inner sep=0.75pt]   
	[align=left] {$\Bit\displaystyle _{0}$};
\draw (304.0, 214.0) node [anchor=north west][inner sep=0.75pt]   
	[align=left] {$\Bit\displaystyle _{n-1}$};
\draw (519, 42.0) node [anchor=north west][inner sep=0.75pt]  
	[font=\huge] [align=left] 
	{\begin{minipage}[lt]{48.28pt}\setlength\topsep{0pt}
			\begin{center}
				$\displaystyle \binom{S}{2}$
				\ \\
				\ \\
				\ \\
				\ \\
				$\displaystyle S$
			\end{center}
			
	\end{minipage}};
\draw (184.0, 19.0) node [anchor=north west][inner sep=0.75pt]   
	[align=left] {$\displaystyle \bot $};
\draw (243.0, 19.0) node [anchor=north west][inner sep=0.75pt]   
	[align=left] {$\displaystyle \top $};
\draw (204.0, 20.0) node [anchor=north west][inner sep=0.75pt]   
	[align=left] {$\displaystyle \alpha $};
\draw (224.0, 20.0) node [anchor=north west][inner sep=0.75pt]   
	[align=left] {$\displaystyle \alpha $};
\draw (66.0, 155.0) node [anchor=north west][inner sep=0.75pt]  
	[font=\Large] [align=left] {$\displaystyle \overline{i}$};
\draw (89.0, 162.0) node [anchor=north west][inner sep=0.75pt]   
	[align=left] {$\Cir\displaystyle _{v_{n-1}}$};
\draw (229.0, 112.0) node [anchor=north west][inner sep=0.75pt]  
	[font=\Huge] [align=left] {$\displaystyle \dotsc $};
\draw (127.0,262.0) node [anchor=north west][inner sep=0.75pt]  
	[font=\Huge] [align=left] {$\displaystyle \dotsc $};
\draw (301.0,262.0) node [anchor=north west][inner sep=0.75pt]  
	[font=\Huge] [align=left] {$\displaystyle \dotsc $};
\draw (63.0,216.0) node [anchor=north west][inner sep=0.75pt]  
	[font=\LARGE] [align=left] {$\displaystyle \dotsc $};
\draw (238.0,216.0) node [anchor=north west][inner sep=0.75pt]  
	[font=\LARGE] [align=left] {$\displaystyle \dotsc $};
\draw (411.0,216.0) node [anchor=north west][inner sep=0.75pt]  
	[font=\LARGE] [align=left] {$\displaystyle \dotsc $};
	
    \end{tikzpicture}

  \endgroup
}

\begin{document}

\maketitle
 
\begin{abstract}
 We revisit the algorithmic problem of reconstructing a graph from
 homomorphism counts that has first been studied in (Böker et al.,
 STACS 2024): given graphs $F_1,\ldots,F_k$ and counts
 $m_1,\ldots,m_k$, decide if there is a graph $G$ such that the
 number of homomorphisms from $F_i$ to $G$ is $m_i$, for all $i$. We
 prove that the problem is NEXP-hard if the counts $m_i$
 are specified in binary and $\Sigma_2^p$-complete if they are in unary.

 Furthermore, as a positive result, we show that the unary version
 can be solved in polynomial time if the constraint graphs are stars
 of bounded size. 
\end{abstract}

\section{Introduction}

Homomorphism counts reveal valuable information about graphs. A
well-known theorem, due to Lovász \cite{Lovasz67}, states that a graph
$G$ can be characterised up to isomorphism by the homomorphism counts
$\hom(F,G)$ from all graphs $F$ into $G$.  In recent years, it
has become increasingly clear that many natural properties of graphs are characterised by homomorphism counts from
restricted graph classes \cite{DellGR18,Dvorak10,Grohe20b,GroheRS25,KarRS025,MancinskaR20,Seppelt24,seppelt_logical_2023}. For example, the homomorphism
counts from all cycles characterise the spectrum of a
graph. Homomorphism counts from trees characterise a graph up to
fractional isomorphism~\cite{Dvorak10,Tinhofer91}, and homomorphism counts from planar
graphs characterise a graph up to quantum
isomorphism~\cite{MancinskaR20}. It has been suggested in
\cite{Grohe20c} that homomorphism counts can be used to 
define vector embeddings of graphs in a systematic and principled way: for every class $\CF$ of graphs we
define a mapping
$G\mapsto\big(\hom(F,G)\big)_{F\in\CF}\in\mathbb R^{\CF}$ mapping each graph $G$ into the (potentially infinite dimensional) vector space $\mathbb R^{\CF}$, the \emph{latent space} of the embedding. We call such
embeddings \emph{homomorphism embeddings}. Vector embeddings of graphs are
mainly of interest in a machine-learning context, because typical ML
algorithms operate on vector representations of the data. Homomorphism
embeddings have been shown to be useful in practice \cite{JinBCL24,WolfOPG23a},
and they are related to other vector embeddings of graphs, such as
those obtained from graph kernels or computed by graph neural networks (see
\cite{Grohe20c}). 

Now suppose we have a vector embedding and carry out computations in
the latent space. For example, we may run an optimisation algorithm 
and find a point in the latent space that is
optimal in some sense. How do we get back a graph from this point?
For homomorphism embeddings, this is the question of how we reconstruct a graph from
homomorphism counts. This is the central problem we study in
this paper; let us state it formally.

\dproblem{\textsc{HomRec}}
{Pairs $(F_1, m_1), \dots, (F_k,m_k)$, where $F_1,\ldots,F_k$
  are graphs and $m_1,\ldots,m_k\in\mathbb N$ (in binary encoding).}
{Is there a graph $G$ such that
	$\hom(F_i, G) = m_i$ for every $i \in [k]$?}

Note that here we are looking at an embedding into a finite-dimensional latent space; the class $\CF$ from the discussion above is $\{F_1,\ldots,F_k\}$ here.We call the pairs $(F_i,m_i)$ \emph{constraints}, the graphs $F_i$
\emph{constraint graphs}, and the numbers $m_i$ the \emph{counts}.       
This problem was first studied systematically
by Böker, Härtel, Runde, Seppelt, and Standke~\cite{BokerHRSS24}. The main results of that paper were
various hardness and a few tractability results for restricted
versions of the problem. For example, the problem is \textsf{NP}-hard
even for constraint graphs of bounded tree width.
The general problem was shown to be in \textsf{NEXP} and hard for the
complexity class $\textsf{NP}^{\textsf{\#P}}$. The exact complexity of
the problem was left open.
Our first theorem settles this question.

\begin{theorem}\label{thm:nexp-complete}
  {\upshape \HomR}\ is \NEXP-complete.
\end{theorem}

The proof of this result is a reduction from the \textsc{SuccinctClique}
problem, known to be \textsf{\upshape NEXP}-complete from
\cite{PapadimitriouY86}.

So far, following~\cite{BokerHRSS24}, we have always assumed the counts $m_i$ in the constraints to
be encoded in binary. However, it may be more natural to encode the
numbers in unary. Let us call the resulting version of the problem
\textsc{UnHomRec}. The reason that \textsc{UnHomRec} may be the more
natural problem is the observation \cite[Lemma~6]{BokerHRSS24} that if an instance $(F_1, m_1), \dots, (F_k,m_k)$ has a solution
satisfying all constraints, then it has a solution of order at most
$\sum_{i=1}^k m_i|F_i|$, simply because we can delete all vertices
not in the image of the homomorphisms required to satisfy the
constraints. If the $m_i$ are encoded in unary, this is polynomial in
the input size. This puts \textsc{UnHomRec} directly into the
complexity class $\textsf{NP}^{\textsf{\#P}}$: we guess a solution and then verify
that all constraints are satisfied using a
$\textsf{\#P}$-oracle. The exact complexity of
\textsc{UnHomRec} is lower, though.

\begin{theorem}\label{thm:sigma2-complete}
  {\upshape \UnHomR}\ is $\Sigma_2^p$-complete.
\end{theorem}

So interestingly, by switching to a succinct encoding, we move the problem from $\Sigma_2^p$ (rather than $\NP$) to $\NEXP$, which is a bit unusual: typically, succinct encodings of problems cause an exponential jump in complexity, from \NP\ to \NEXP\ \cite{PapadimitriouY86}, but here we only jump from $\Sigma_2^p$ to \NEXP.

Having now settled the exact complexity of the homomorphism
reconstruction problem, both in its unary and binary versions, we set
out to look for tractable special cases.  For the
binary version, B\"oker et al.~\cite{BokerHRSS24} proved that reconstruction can be
\textsf{NP}-hard even for a fixed set of constraint graphs that are all labelled trees.  For the unary version, it is not clear if this can happen as
well. But even for very simple special cases, it is not at all obvious
how we can solve the reconstruction problem efficiently. Consider the case with
two constraints $(\smallKone,n)$,
$(\smallPtwo,m)$, where the constraint graphs are a single vertex
graph and the path of length $2$.  Intuitively, this is the problem of
deciding if there is a graph $G$ with $n$ vertices and $m$ homomorphic
images of a path of length $2$ in time polynomial in $n$ and
$m$. Arguably, the subgraph version of this problem is more natural: given $n$ and $m$,
decide if there is a graph with $n$ vertices and $m$ paths of length $2$
as subgraphs. It is not obvious how to construct such a
graph in time polynomial in $n$ and $m$. (This is the simplest case
that was left open in \cite{BokerHRSS24}.) The following theorem
implies that there is a polynomial-time algorithm.

Recall that a \emph{star} is a tree of
height $1$, that is, a connected graph that has at most one vertex of
degree greater than $1$. Both the $1$-vertex graph and the path of
length $2$ are examples of stars.

\begin{theorem}\label{theo:star}
  The following problem is solvable in
  time $m^{O(\ell^2)}$, where $\ell\coloneqq\max\{|F_i|\mid i\in[k]\}$ and $m=\max\{m_i\mid i\in[k]\}$.
  \dproblem{\textsc{StarHomRec}}{$(F_1, m_1), \dots, (F_k,m_k)$, where $F_1,\ldots,F_k$
    are stars and $m_1,\ldots,m_k\in\mathbb N$ (in unary
    encoding)}{Is there a graph $G$ such that $\hom(F_i,G)=m_i$ for
    all $i\in[k]$?}

\end{theorem}

To prove Theorem~\ref{theo:star}, we observe that the homomorphism counts
$\hom(S,G)$ from stars $S$ into a graph $G$ only depend on the degree
sequence of $G$. Based on this observation, we devise a dynamic-programming algorithm  that computes a feasible degree sequence and then reconstructs a graph from this degree sequence using a known
algorithm due to Havel~\cite{Havel1955} and Hakimi \cite{Hakimi1962}.

Noting that the
homomorphism counts from all stars of order at most $\ell$ determine the
subgraph counts for all stars of order at most $\ell$ (see
\cite{CurticapeanDM17,Lovasz12}), as a corollary, we obtain the
corresponding result for subgraph counts. In fact, it will be easier to prove the
version for subgraph counts first and then derive the result for homomorphism counts.

\begin{corollary}\label{cor:star}
  The following problem is solvable in
  time $m^{O(\ell^2)}$, where $\ell\coloneqq\max_{i\in[k]}|F_i|$ and $m=\max_{i\in[k]}m_i$.

  \dproblem{\textsc{StarSubRec}}{$(F_1, m_1), \dots, (F_k,m_k)$, where $F_1,\ldots,F_k$
    are stars and $m_1,\ldots,m_k\in\mathbb N$ (in unary
    encoding)}{Is there a graph $G$ that
    has $m_i$ subgraphs isomorphic to $F_i$, for every $i\in[k]$?}

\end{corollary}

We note that for both Theorem~\ref{theo:star} and
Corollary~\ref{cor:star}, we can not only solve the decision problem, but actually construct a graph satisfying the constraints if there
exists one.

\subsection*{Related Work}
Homomorphisms and also homomorphism counts play an important role in many areas of mathematics and computer science and have also been studied in complexity theory. Leading to this paper is the recent work on homomorphism indistinguishability already discussed at the beginning of the introduction. Let us also mention work on the complexity of homomorphism reconstruction, from which we borrow some ideas \cite{boker_complexity_2019,CernyS26,Seppelt24}.

The question of whether we can reconstruct a structure from substructures or substructure counts has a long tradition. Famously, the Ulam-Kelly Reconstruction Conjecture says that a graph of order $n$ is uniquely determined by the multiset of its subgraphs of order $n-1$ \cite{Kelly57,Ulam60} (also see \cite{Babai95,oneil_ulam_1970}).
While the reconstruction conjecture is not algorithmic, there has been some work on the complexity of reconstructing graphs from small patterns \cite{ko_three_1991}. 

Closely related to homomorphism reconstruction is the question of whether a graph can realise certain homomorphism densities, which is studied in extremal graph theory (e.g. \cite{razborov_minimal_2008}). An interesting undecidability result for reconstruction from homomorphism densities has been proved in \cite{hatami_undecidability_2011}. 

In database theory, homomorphisms have played an important role since Chandra and Merlin's \cite{ChandraM77} characterisation of the containment problem for conjunctive queries in terms of homomorphisms. In a similar way, homomorphism counts are related to the containment problem for conjunctive queries under bag semantics \cite{chaudhuri_optimization_1993}. Indeed, this problem is equivalent to the problem of deciding if, for two relational structures $A,B$, there is a structure $X$ such that $\hom(A,X)>\hom(B,X)$. It is a long-standing open problem whether this is decidable (see \cite{kopparty_homomorphism_2010,MarcinkowskiO24,MarcinkowskiO25} for partial results).

\subsection*{Organisation of the Paper}
We prove \cref{thm:nexp-complete} in Section~\ref{sec:binary}, \cref{thm:sigma2-complete} in Section~\ref{sec:unary}, and \cref{theo:star} in Section~\ref{sec:algorithms}. 

\section{Preliminaries}\label{sec:prel}

We write $\mathbb{N} = \{0,1,2,\dots\}$ for the set of natural numbers, and we let $[k,\ell] \coloneqq \{k,\dots,\ell\}\subseteq \NN$ and $[k]\coloneqq[1,k]$.
$\DecProblemA \le_p \DecProblemB$ denotes that
a decision problem $\DecProblemA$ is polynomial-time
many-one reducible to the decision problem $\DecProblemB$.
We denote the vertex set of a \emph{graph} or directed acyclic graph (\emph{dag}) $G$ by $V(G)$ and the edge set by $E(G)$.
For ease of notation, we denote an edge by $uv$ or $vu$.
A \emph{homomorphism} from a graph $F$ to a graph $G$ is a mapping
$h \colon V(F) \to V(G)$ such that $h(uv) \in E(G)$ for every $uv \in E(F)$.
A \emph{($\CC$-vertex-)coloured graph} is a triple $G = (V, E, c)$ where $(V, E)$ is a graph,
the \emph{underlying graph},
and $c \colon V(G) \to \CC$ a function assigning a \emph{colour} from
a set $\CC$ to every vertex of $G$. For $\mathsf{C}\in \CC$, we let $\mathsf{C}^G \coloneqq c^{-1}(\mathsf{C})$ be the \emph{colour class} of $\mathsf{C}$ in $G$.
An \emph{($\CL$-)labelled} graph is defined analogously
with a function $\ell \colon \CL \to V(G)$ assigning a vertex of~$G$
to every \emph{label} from a set of labels $\CL$ instead.
Homomorphisms between coloured graphs and between labelled graphs are then
defined as homomorphisms of the underlying graphs that
respect colours and labels, respectively.

A graph $G'$ is a \emph{subgraph} of a graph $G$, written $G' \subseteq G$,
if $V(G') \subseteq V(G)$ and $E(G') \subseteq E(G)$.
The \emph{subgraph induced by a set} $U \subseteq V(G)$,
written $G[U]$, is the
subgraph of $G$ with vertices $U$ and edges $E(G) \cap U^{2}$.
We write $\hom(F, G)$ for the number of homomorphisms from $F$ to $G$ and
$\sub(F, G)$ for the number of subgraphs $G' \subseteq G$ such that $G' \cong F$.
This notation generalises to coloured and labelled graphs
in the straightforward way.  For a vertex $v$ in an undirected graph $G$, we let $N(v)\coloneqq\{w\in V(G) \mid vw\in E(G)\}$ be the \emph{neighbourhood} of $v$, and we let $\deg(v)\coloneqq\abs{N(v)}$ be the \emph{degree} of $v$. For a vertex $v$ in a dag $G$, we let $N^+(v)\coloneqq\{w\in V(G) \mid vw\in E(G)\}$ be the \emph{out-neighbourhood} of $v$ and $N^-(v)\coloneqq\{w\in V(G) \mid wv\in E(G)\}$ the \emph{in-neighbourhood} of $v$; we let $\deg^+(v)\coloneqq\abs{N^+(v)}$ be the \emph{out-degree} and $\deg^-(v)\coloneqq\abs{N^-(v)}$ the \emph{in-degree}. We call nodes of in-degree $0$ \emph{sources} and nodes of out-degree $0$ \emph{sinks}.

The \emph{degree sequence} of a graph is the sequence of its vertex degrees, ordered decreasingly. A sequence $\vec d=(d_1,\ldots,d_n)$ of nonnegative integers I \emph{graphic} if there is a graph $G$ with degree sequence $\vec d$. For example, the sequence $(3,3,2,2,0)$ is graphic and the sequence $(3,3,1)$ is not. Note that if there is a graph $G$ with degree sequence $(d_1,\ldots,d_n)$ then $|V(G)|=n$ and $|E(G)|=\frac{1}{2}\sum_{i=1}^nd_i$.

A \emph{boolean circuit} $C$ is a dag where nodes are coloured with $\{\Inp, \lor,\land ,\neg\}$. Nodes coloured $\lor$ or $\land$ have in-degree $2$, the $\neg$ nodes in-degree 1, and we call them \emph{gates}. $\Inp$ nodes $x_1,\dots,x_n$ are sources that we call \emph{inputs}, and sinks $y_1,\dots,y_m$ \emph{outputs}. A circuit $C$ computes a function $C : \{0,1\}^n \to \{0,1\}^m$. The size of $C$, denoted by $\lvert C \rvert$, is the number of nodes in it.

\section{\NEXP-Completeness of Binary Homomorphism Reconstruction}
\label{sec:binary}

In this section we prove \cref{thm:nexp-complete}. Whenever we use coloured graphs, denoted $F$ or $G$, we will not explicitly state their colouring function $c:V(F)\to \CC$. The set of colours $\CC$ contains symbols such as $\PPP$, $\alpha$, or $\bot$, and abbreviations that begin with a capital letter such as $\Cir_i$ or $\Bit_i$, whenever we need to parametrise colour classes with an index $i$.

Recall that the \emph{constraints} of a \textsc{HomRec} instance are pairs $(F,n)$ where $F$ is a graph and $n$ a nonnegative integer in binary representation. We denote the concatenation of lists of constraints as union, and we make no distinction between individual constraints and lists of length one. Note that the order in which the constraints appear in a \HomR-instance does not matter.

The missing proofs of \cref{claim:3.1-correct}, \cref{claim:3.2-correct}, and \cref{claim:m:4-correct} can be found in the appendix.

\subsection{Transforming Boolean Circuits into Homomorphism Constraints}\label{sec:npc}

As a warm-up result for proving that \HomR\ is \NEXP-complete, we will show that \textsc{ColHomRec}, where the constraint graphs $F_i$ (and the target graph $G$) may be coloured graphs, is \NP-hard. While this is known \cite{BokerHRSS24}, our new proof provides us with a good opportunity to introduce the main technique of \cref{sec:nexp} -- representing boolean circuits in the graph we are aiming to reconstruct.

\dproblem{\CircuitSAT}
{Circuit $C:\{0,1\}^n\to\{0,1\}$ where $n\in\NN$.}
{Is there an input $x \in \{0,1\}^n$ such that $C(x)=1$?}
    
The coloured graph $G=(V(G),E(G),c:V(G)\to\CC)$ we will reconstruct, if an input $x$ with $C(x)=1$ exists, should contain an explicit representation of the computation of $C$ on $x$. Gates will correspond to nodes, truth values will be represented as unique vertices of colours $\bot$ and $\top$, and edges from nodes encode a truth value depending on whether they connect them with $\bot$ or $\top$. 

\begin{theorem}\label{thm:circuitsat}
    ~\hfill$\CircuitSAT \leq_p \textsc{\upshape ColHomRec}$.\hfill~
\end{theorem}

\emph{The $n:m$ constraint.} In the construction of homomorphism constraints for representing the evaluation of boolean circuits, we will make use of constraints that enforce regularity between certain colour classes. By calculating appropriate multiplicities of their occurrences, we can let homomorphism counts from stars -- rather, sums of powers of degrees -- do the rest. For any two colours $\AAA,\BBB\in \CC$ and positive integers $n,m\in\NN$, we define the following list of constraints:
\[ \Fnm(\AAA,\BBB,n,m) \coloneqq \big( (\smallKone \AAA, n), (\AAA \smallKtwo \BBB, nm),(\smallPtwoLabelled{\BBB}{\AAA}{\BBB}, nm^2)\big ). 
\] 
\begin{lemma}\label{lemma:stars}
    Any graph $G$ that satisfies $\Fnm(\AAA,\BBB,n,m)$ contains a colour class $\AAA^G$ of size $n$. Each vertex coloured $\AAA$ has exactly $m$ neighbours coloured $\BBB$.
\end{lemma}
\begin{proof}
    Let $G$ be some graph satisfying the constraints. Then its colour class $\AAA^G$ has size $\abs{\AAA^G}=n$. Denote $\deg_\BBB(v)\coloneqq \abs{N(v)\cap \BBB^G}$. Observe that for the number of homomorphisms to $G$ holds $\hom(\AAA \smallKtwo \BBB, G) = \sum_{v\in \AAA^G} \deg_\BBB(v) = nm$. Furthermore, we have $\hom(\smallPtwoLabelled{\BBB}{\AAA}{\BBB}, G) = \sum_{v\in \AAA^G} \deg_\BBB(v)^2 = n m^2$. By applying Cauchy–Schwarz it follows:
    \[
    n m^2 = \sum_{v\in \AAA^G} \deg_\BBB(v)^2 \;\ge\; \frac{1}{n}\Big(\sum_{v\in \AAA^G} \deg_\BBB(v)\Big)^2
    = \frac{1}{n}(n m)^2 = n m^2.
    \]
    Equality in the Cauchy--Schwarz inequality holds if and only if for $v\in \AAA^G$ all $\deg_\BBB(v)$ are equal, therefore every vertex of colour $\AAA$ has exactly $m$ neighbours in $\BBB$.
\end{proof}

In order to create graphs that mirror the combinatorial structure of boolean circuits, we have to come up with a technique to represent different types of gates and logical truth values. We can greatly simplify this problem by imposing an order on the nodes in $V(C)$ and colouring them individually. This way, we can introduce homomorphism constraints that ensure that, in all occurrences of a reconstructed copy of the circuit $C$, locally the computation at each gate is correct. Truth values are represented as nodes being adjacent to either one of two vertices in the \emph{value gadget}. For this gadget, which is unique and globally controls the evaluation at each gate, we reserve the colours $\{\alpha, \bot, \top\}$.

Let the \emph{value gadget} $G_\alpha$ be the coloured graph $\smallAlpha$, and let
\[
    \cons{F_\alpha} \coloneqq \big ( (\smallKone \alpha,2), (\smallKone \bot,1), (\smallKone \top,1), (\smallAlpha,1), (\alpha\smallKtwo\bot,1), (\alpha\smallKtwo\top,1)\big ).
\]

Furthermore, for every colour $\AAA$, let 
\[ 
    \cons{F_{\left\llbracket \cdot \right\rrbracket}}(\AAA,n) \coloneqq \Fnm(\AAA, \alpha, n, 1).
\]
If $G$ satisfies the constraints $\cons{F_\alpha} \cupp \cons{F_{\val \cdot}}(\AAA,k)$ for some $k\in\NN$ and $\AAA \in \CC$, we can define the value function $\val{v}_\AAA:\AAA^G\to\{0,1\}$
by letting $\val{v}_\AAA=1$, if $v$ is adjacent to the $\alpha$ vertex next to $\top$, and $\val{v}_\AAA=0$, otherwise.
We will simply write $\val{v}$, if $G$ is clear from context and the colour of $v$ irrelevant.

\begin{proof}[Proof of \cref{thm:circuitsat}.]
    Let $C = (V(C),E(C),\Inp^C,\lor^C,\land^C,\neg^C)$ be a circuit with $n$ input gates $v_1,\dots,v_n$ and $1$ output gate $o$.

    \medskip
    \emph{The Graph $G_{C(x)}$.}
For input $x\in\{0,1\}^n$, we define an undirected coloured graph $G_{C(x)}$ as follows: We start with a value gadget. Then for every vertex $v \in V(C)$, we introduce a vertex with colour $\Cir_v$. Depending on the value of $v$ in $C(x)$, the vertex $v$ of $G_{C(x)}$ is connected to the $\alpha$-vertex next to $\bot$ if the value is $0$ or to the $\alpha$-vertex next to $\top$ if the value is $1$.

    \emph{The Homomorphism Constraints.}
For gate $v\in \lor^C$ we write $v=\lor(u,w)$, if $v$ has incoming edges $uv, wv \in E(C)$. Similarly, we write $v=\land(u,w)$ if $v \in \land^C$, or $\neg(w)$ if $v \in \neg^C$, for the other two types of gates. We define the constraints $\cons{F}_C(v)$ as follows:
    
    \[
        \cons{F}_C(v) \coloneqq
        \begin{cases}
          \big((\smallAlphaAnd{0}{0}{1},0), (\smallAlphaAnd{0}{1}{1},0), (\smallAlphaAnd{1}{0}{1},0), (\smallAlphaAnd{1}{1}{0},0)\big), & \text{$v = \land(u,w)$}, \\[6pt]
          \big((\smallAlphaOr{0}{0}{1},0), (\smallAlphaOr{0}{1}{0},0), (\smallAlphaOr{1}{0}{0},0), (\smallAlphaOr{1}{1}{0},0)\big), & \text{$v = \lor(u,w)$}, \\[6pt]
          \big((\smallAlphaNeg{0}{0},0), (\smallAlphaNeg{1}{1},0)\big), & \text{$v = \neg(w)$}.
        \end{cases}
    \]
    
    We construct the following list of constraints $\cons{F}$ such that $\exists x \in \{0,1\}^n:C(x)=1\Leftrightarrow \cons {F} \in \textsc{HomRec}$:
    
    \[ \cons{F} \coloneqq \cons{F_\alpha} \cupp \big ((\smallPtwoOut, 1)\big ) \cupp \bigcup_{v \in V(C)} \cons{F_{\val \cdot}}(\Cir_v, 1) \cupp \bigcup_{v \in V(C)\setminus \Inp^C} \cons{F}_C(v) . \]

   If some $x\in\{0,1\}^n$ with $C(x)=1$ exists, we observe that $G_{C(x)}$ as defined before satisfies all constraints in $\cons{F}$.
    \begin{claim}\label{claim:3.1-correct}
        If there is a graph $G$ satisfying all constraints in $\cons {F}$, there exists an $x\in\{0,1\}^n$ such that $C(x)=1$. Furthermore, $\cons {F}$ can be computed in polynomial time.
    \end{claim}
    This completes the proof of \cref{thm:circuitsat}.
\end{proof}

\subsection{\NEXP-Completeness of Coloured Homomorphism Reconstruction}\label{sec:nexp}

\NP-complete problems with exponentially succinct representations are often complete for nondeterministic exponential time \NEXP\ \cite{PapadimitriouY86}. The notion of succinct input representations was first investigated by Galperin and Wigderson in \cite{GalperinW83}. For example, a graph can be succinctly encoded by a boolean circuit as follows.
We denote the binary encoding of a nonnegative integer $i\in\NN$ by $\bin i$.

\begin{definition}[\cite{GalperinW83}]
    Let $G$ be a graph with $m\leq2^n$ vertices $v_0,\dots,v_{m-1}$. A \emph{succinct representation (SCR)} of a graph $G$ is a boolean circuit  $C_G$ with $2n$ inputs and $1$ output such that $C_G(\bin{i},\bin{j})=1$ if $i,j<m$ and $(v_i,v_j)\in E(G)$ and $C_G(\bin{i},\bin{j})=0$ otherwise.
\end{definition}

We reduce from the following $\SuccinctClique$-problem, which is \NEXP-complete \cite{GalperinW83}.

\dproblem{\SuccinctClique}
{SCR $C_G$, positive integer $k\in\NN$ (in binary).}
{Does $G$ contain a $k$-clique?}
As the first and most difficult step in the proof of \cref{thm:nexp-complete}, the following lemma shows that \textsc{ColHomRec} is \NEXP-hard.

\begin{lemma}\label{lem:nexphard}
    ~\hfill$\SuccinctClique \leq_p \textsc{\upshape ColHomRec}$.\hfill~
\end{lemma}

\emph{A Technical Change to the Encoding of Boolean Circuits from \cref{sec:npc}.}  We consider an instance $(C_G,k)$ of $\SuccinctClique$, where $C_G$ is the SCR of a graph $G$. We let $o$ be the output of $C_G$. Note that the circuit $C_G$ has $2n$ inputs $v_0,\dots,v_{2n-1}$ (rather than $n$ as the circuits we considered before); we have to adapt our construction accordingly. Moreover, we will be creating $N\coloneqq k(k-1)/2$ copies of the representation of the circuit $C_G$. For consistency, we keep only one value gadget. This means that our constraints need to change slightly.

We add constraints $\cons{F}_\PPP$ that will produce $N$ vertices coloured $\PPP$, one for each copy, which are connected to every vertex in their respective copies. The idea is that the $\PPP$ nodes together with their representation of $C_G$ will form disjoint stars that can be evaluated independently from each other. We would prefer to individualize each $C_G$ with its completely own set of colours $\Cir_v^{(1)},\dots,\Cir_v^{(k)}$ instead, but have to avoid introducing $\Omega(k)$ constraints. Let
\[
    \cons{F}_\PPP \coloneqq \bigcup_{v\in V(C_G)} \Fnm(\PPP,\Cir_v,N,1) \cupp \bigcup_{v\in V(C_G)} \Fnm(\Cir_v,\mathsf \PPP,N,1).
\]
$\cons{F}_\PPP$ forces the colour classes $\PPP$ and $\Cir_v$ to have $N$ vertices. Moreover, it forces each vertex in $\PPP$ to have $1$ neighbour in $\Cir_v$ and each vertex in $\Cir_v$ to have $1$ neighbour in $\PPP$. The constraints $\cons{F}_{C_G}(v)$ enforcing correct behaviour of Boolean connectives are made less restrictive, by limiting them to $\Cir_u,\Cir_w$, and $\Cir_v$ nodes adjacent to the same $\PPP$ vertex:
 
\[
    \cons{F}_{C_G}(v) \coloneqq
    \begin{cases}
      \big((\smallAlphaAndX{0}{0}{1},0), (\smallAlphaAndX{0}{1}{1},0), (\smallAlphaAndX{1}{0}{1},0), (\smallAlphaAndX{1}{1}{0},0)\big), & \text{$v = \land(u,w)$}, \\[6pt]
      \big((\smallAlphaOrX{0}{0}{1},0), (\smallAlphaOrX{0}{1}{0},0), (\smallAlphaOrX{1}{0}{0},0), (\smallAlphaOrX{1}{1}{0},0)\big), & \text{$v = \lor(u,w)$}, \\[6pt]
      \big((\smallAlphaNegX{0}{0},0), (\smallAlphaNegX{1}{1},0)\big), & \text{$v = \neg(w)$}.
    \end{cases}
\]

So far, we only guarantee the existence of $\binom{k}{2}$ copies of the SCR $C_G$ in the graph we are trying to reconstruct. We need additional constraints $\cons{F}_\QQQ, \cons{F}_{\text{in}}$, and $\cons{F}_{\text{str}}$ to ensure that not only are there $\binom{k}{2}$ edges in $G$, but that $k$ mutually adjacent vertices $S\subseteq V(G)$ forming a $k$-clique exist. The vertices in $S$ will be represented as $k$ node gadgets encoding their respective indices in $G$. The index of each node gadget (in binary representation) is fed into $k-1$ copies of $C_G$, which will have to be correctly evaluated by any graph satisfying our list of constraints. For node gadgets, we introduce the colours $\Bit_i$, for $i\in[0,n-1]$, and $\QQQ$. We denote the first and second set of $n$ inputs to $C_G$ as $\Inp_1\coloneqq\{v_0,\dots,v_{n-1}\}\subseteq V(C_G)$ and $\Inp_2\coloneqq\{v_n,\dots,v_{2n-1}\}\subseteq V(C_G)$, respectively, and as  index $\text{idx}(v_i)\coloneqq i$. Let
\begin{align*}
    \cons{F}_{\QQQ} \coloneqq &\bigcup_{i\in[0,n-1]} \cons{F}_{\val \cdot}(\Bit_i, k) \cupp \bigcup_{i\in[0,n-1]}\Fnm(\QQQ,\Bit_i,k,1) \cupp \bigcup_{i\in[0,n-1]}\Fnm(\Bit_i,\QQQ,k,1), \\
    \cons{F}_{\text{in}} \coloneqq &\bigcup_{v \in \Inp_1} \Fnm(\Cir_v,\Bit_{\text{idx}(v)},N,1) \cupp \bigcup_{v \in \Inp_2}\Fnm(\Cir_v,\Bit_{\text{idx}(v)-n},N,1) \\
    \cupp &\bigcup_{v \in \Inp_1} (\smallAlphaCfour{\text{idx}(v)}, 0) \cupp \bigcup_{v \in \Inp_2} (\smallAlphaCfour{\text{idx}(v)-n}, 0).
\end{align*}
\emph{Combining $n:m$ Constraints Guarantees Consistent Circuit Inputs.} The constraints in $\cons{F}_{\QQQ}$ enforce colours classes $\QQQ$ and $\Bit_i$ of size $k$. Again, vertices coloured $\QQQ$ are required to have $1$ neighbour coloured $\Bit_i$, and vice versa. Furthermore, $\Bit_i$ has to define the value function $\val{\cdot}$.

For inputs $v\in \Inp^{C_G}$, $\cons{F}_{\text{in}}$ enforces that each vertex coloured $\Cir_v$ has exactly one neighbour coloured $\Bit_{\text{idx(v)}}$ or $\Bit_{\text{idx(v)}-n}$, 
depending on whether $v\in \Inp_1$ or $v\in \Inp_2$ holds. Here, we do not require the converse, since we would like every vertex coloured $\Bit_{i}$ to have $k-1$ neighbours total in $\Cir_{v_i}$ and $\Cir_{v_{i+n}}$. The third row of constraints enforces the values $\val{\cdot}$ of vertices in $\Cir_v$ and their neighbours in $\Bit_{\text{idx(v)}}$ or $\Bit_{\text{idx(v)}-n}$ to be equal.

To make the indices of node gadgets behave as input strings of length $n$, we introduce the two constraints shown in \cref{fig:diamond-gadgets}, and let
\[
    \cons{F}_{\text{str}} \coloneqq \big ( (F_{\text{str}_1}, N), (F_{\text{str}_2}, N) \big ).
\]

\begin{figure}[ht]
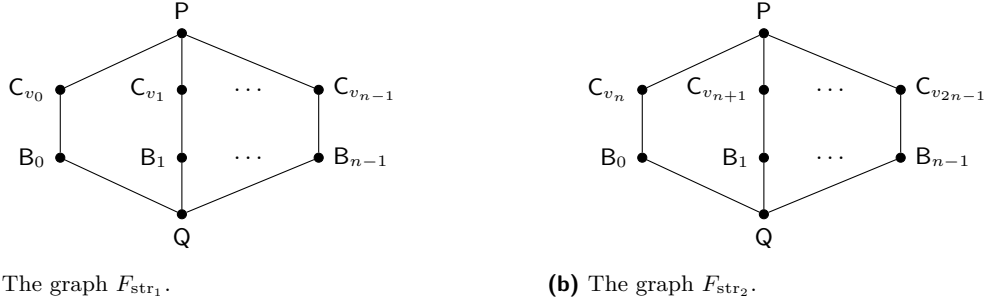

  \centering
  \begin{minipage}{0.45\textwidth}
    \centering
    $\diamondGadgetA$ \subcaption{The graph $F_{\text{str}_1}$.}
    \label{fig:diamond0}
  \end{minipage}\hfill
  \begin{minipage}{0.45\textwidth}
    \centering
    $\diamondGadgetB$ \subcaption{The graph $F_{\text{str}_2}$.}
    \label{fig:diamondn}
  \end{minipage}
  \caption{These constraint graphs are the only configurations between inputs and vertices encoding indices we want to allow. Recall that the set of $2n$ inputs to $C_G(\bin{i},\bin{j})$ is $v_0,\dots,v_{n-1}, v_n,\dots,v_{2n-1}$.}
  \label{fig:diamond-gadgets}
\end{figure}

\begin{proof}[Proof of \cref{lem:nexphard}.]
    We combine the previous constraints into $\cons{F^+} \coloneqq \cons{F}_\PPP \cupp \cons{F}_{\QQQ} \cupp \cons{F}_{\text{in}} \cupp \cons{F}_{\text{str}}.$ We will show that for any SCR $C_G$ and $k\in\NN$ holds that $G$ contains mutually adjacent $S=\{s_1,\dots,s_k\} \Leftrightarrow \cons{F}\in \textsc{HomRec}$, for the following list of constraints where $N \coloneqq k(k-1)/2$:
    \[ \cons{F} \coloneqq \cons{F}_\alpha \cupp \big ((\smallPtwoOut,N)\big ) \cupp \bigcup_{v \in V(C_G)} \cons{F}_{\val \cdot}(\Cir_v, N) \cupp \bigcup_{v \in V(C_G)\setminus \Inp} \cons{F}_{C_G}(v) \cupp \cons{F}^+. \]

\emph{The graph $G_S$.} How does a graph satisfying the constraints $\cons{F}$ look like? We have a good idea about the first part; it needs to contain $k(k-1)/2$ copies of the SCR $C_G$, which all evaluate to $1$. But the second half of our constraints, $\CF^+$, interacts with the inputs to our circuits $C_G$.

\begin{figure}[ht]
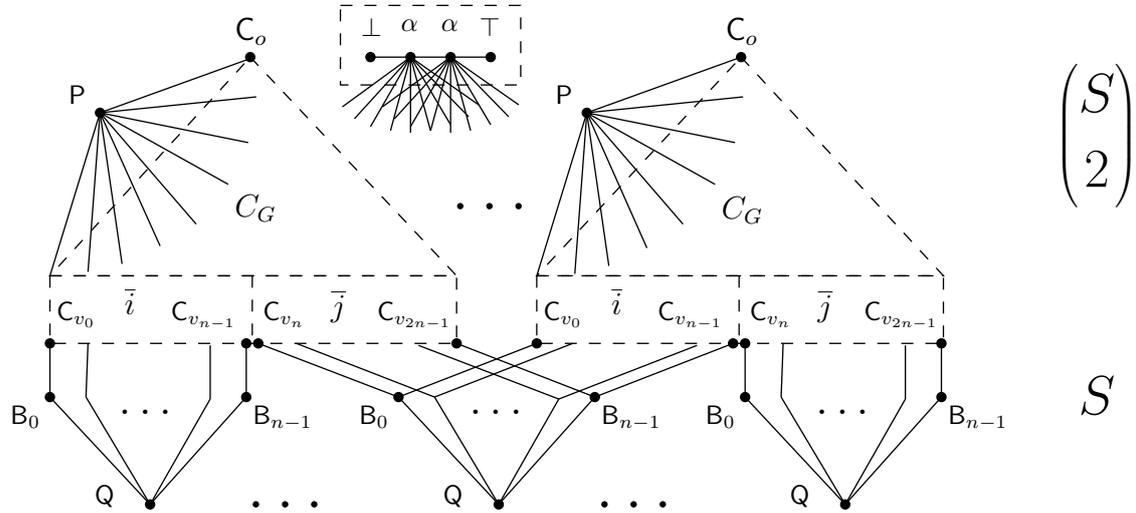

  \centering
  \begin{minipage}{\textwidth}
  \tikzset{every picture/.style={line width=0.5pt}}
    \nexpCircuitGraph
  \end{minipage}
  \caption{The graph $G_S$. If the gadgets for $S$ encode a $k$-clique of $G$, all vertices coloured $\Cir_o$ are adjacent to the ``true'' value gadget vertex next to $\top$.}
  \label{fig:g-s}
\end{figure}

For input SCR $C_G$, we define an undirected coloured graph $G_S$ (see \cref{fig:g-s}) for any size $k$ subset of vertices $S\subseteq V(G)$ as follows:
\begin{itemize}
    \item Let $S=\{s_1,\dots,s_k\}$, and denote their binary encodings with $\bin{s_i}\in\{0,1\}^n$.
    \item Recall that $G_\alpha$ denotes the value gadget \smallAlpha. Denote its vertices, left to right, as $\{v_\bot,v_{\alpha\bot},v_{\alpha\top},v_\top\}$.
    \item Let $V_{\text{bit}} \coloneqq \{b_0,\dots,b_{n-1}\}$.
    \item We will be using $\binom{S}{2} = \{s_1s_2,\dots,s_{k-1}s_k\}$ to index the set $[1,N]$ = $[1,k(k-1)/2]$.
    \item The colours of $G_S$ are $\CC \coloneqq \{\alpha, \top, \bot, \PPP, \QQQ\} \cupp \{C_v\mid v\in V(C_G)\} \cupp \{\Bit_i\mid i\in [0,n-1]\}$.
    \item Let $V(G_S) \coloneqq V(G_\alpha) \cupp \big(\{v_\PPP\}\times \binom{S}{2} \big) \cupp \big(V(C_G)\times \binom{S}{2}\big) \cupp \big(\{v_\QQQ\}\times S\big) \cupp \big(V_{\text{bit}}\times S\big)$.
    \item We let its colour classes (omitting $G_S$) $\alpha\coloneqq \{v_{\alpha\bot},v_{\alpha\top}\}$, $\bot\coloneqq \{v_{\bot}\}$, and $\top\coloneqq \{v_{\top}\}$;
    \item $P \coloneqq \{v_\PPP\}\times\binom{S}{2}$, $Q \coloneqq \{v_\QQQ\}\times S$;
    \item for $v\in V(C_G)$, $\Cir_v\coloneqq \{v\}\times\binom{S}{2}$; and, for $i\in [0,n-1]$, $\Bit_i\coloneqq \{v_i\}\times S$.
    \item We define its edges $E(G_S)$ as follows: any edges from $E(G_\alpha)$.
    \item For each $s_is_j\in \binom{S}{2}$, the vertex $(v_\PPP,s_is_j)$ coloured $\PPP$ has one neighbour $(v,s_is_j)$ with colour $\Cir_v$ for each $v\in V(C_G)$.
    \item For each $s\in {S}$, the vertex $(v_\QQQ,s)$ coloured $\QQQ$ has one neighbour $(b_i,s)$ with colour $\Bit_i$ for each $i\in [0,n-1]$.
    \item For each $i\in [0,n-1]$ and $s\in {S}$, the vertex $(b_i,s)$ is adjacent to $(v,s_is_j)$ if $v$ is the $i$-th input of $C_G$ and $s = s_i$ OR if $v$ is the $n+i$-th input of $C_G$ and $s = s_j$, and
    \item the vertex $(b_i,s)$ is adjacent to $v_{\alpha\bot}$, if $\bin{s}_i=0$, and adjacent to $v_{\alpha\top}$, if $\bin{s}_i=1$.
    \item Observing that this defines $\val{\cdot}$ on the inputs of all copies of $C_G$, we can define the remaining values of $C_G$ copies inductively, by adding edges from gates to either $v_{\alpha\bot}$ or $v_{\alpha\top}$, as in the proof of \cref{thm:circuitsat}.
\end{itemize}  
If the graph $G$ represented by SCR $C_G$ contains some $k$-clique $S=\{s_1,\dots,s_k\}$, then we argue that the graph $G_S$ satisfies our constraints.

\begin{claim}\label{claim:3.2-correct}
    If there is a graph $G$ satisfying all constraints in $\cons {F}$, there exist $k$ strings $\bin{i_1},\dots,\bin{i_k}\in\{0,1\}^n$ such that pairwise $C_G(\bin{i},\bin{j})=1$, and thus, there are $k$ mutually adjacent vertices in the graph encoded by SCR $C_G$. Furthermore, $\cons {F}$ can be computed in polynomial time in $\abs{C_G}+\log k$.
\end{claim}
\cref{claim:3.2-correct} completes the proof of \cref{lem:nexphard}.
\end{proof}

\subsection{Hardness for Uncoloured Graphs}

To prove \cref{thm:nexp-complete}, we will show that for any instance of \textsc{ColHomRec} we can compute an equivalent instance of \textsc{HomRec} where all constraint graphs are simple uncoloured graphs in polynomial time. 

\begin{lemma}\label{lemma:m-to-4}
    There exists a polynomial time reduction from {\upshape \textsc{ColHomRec}}\ (with an unbounded number of colour classes) to {\upshape \textsc{ColHomRec}}\ with at most $4$ colour classes.
\end{lemma}

Let the \emph{$m$-colour gadget} $G_m$ be the $3$-colour graph $\smallColour{m}$, and let
\[
    \cons{F}_X \coloneqq \big ( (\smallKone \XXX,m), (\smallKone \SSS,1), (\smallKone \TTT,1), (\XXX\smallKtwo \SSS,1), (\XXX\smallKtwo \XXX,2m-2), (\XXX\smallKtwo \TTT,1), (G_m,1), \bigcup_{k<m} (G_k,0)\big ).
\]
Observe that the $m$-colour gadget $G_m$ satisfies $\cons{F}_\XXX$. Any other graph $G$ that satisfies $\cons{F}_\XXX$ contains a single copy of $G_m$ as subgraph. Furthermore, $G$ cannot contain other vertices coloured $\SSS$, $\TTT$, or $\XXX$.

\begin{proof}[Proof of \cref{lemma:m-to-4}.]
    From an $m$-colour constraint graph $F=(V(F),E(F),\AAA_1^F,\dots,\AAA_m^F)$ we compute a $4$-colour graph $F^*=(V(F^*),E({F^*}),\SSS^{F^*},\TTT^{F^*},\XXX^{F^*},\AAA^{F^*})$ as follows:
    \begin{itemize}
        \item $V(F^*) \coloneqq V(F) \cupp \{s_S,s_1,\dots,s_m,s_T\},$
        \item $E(F^*) \coloneqq E(F) \cupp \bigcup_{i\in [m]} \{vs_i \mid v\in \AAA_i^F\}\cupp \{s_is_{i+1} \mid i\in[m-1]\}\cupp\{s_Ss_1,s_ms_T\}$,
        \item $\SSS^{F^*} \coloneqq \{s_S\},\; \TTT^{F^*} \coloneqq \{s_T\}, \; \XXX^{F^*} \coloneqq \{s_i\mid i\in[m-1]\}$, and $\AAA^{F^*} \coloneqq V(F)$.
    \end{itemize}

    The graph $F^*$ is the union of $F$ and $G_m$, replacing any colour from $F$ with colour $\AAA$ and an edge to $G_m$. Observe that the example shown in \cref{fig:four-graphs} preserves the number of homomorphisms for each pair of graphs. In the following proof of \cref{claim:4-colour-encoding}, we will make this fact explicit.

    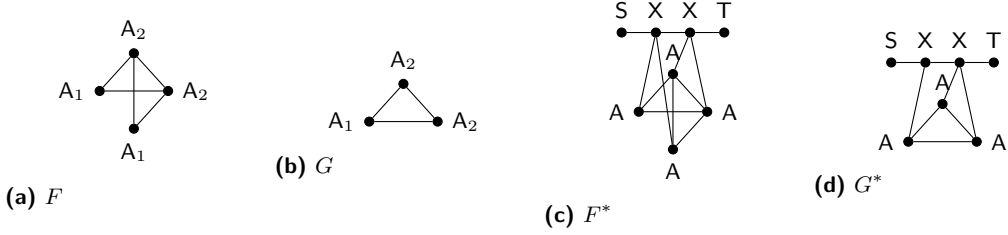
\begin{figure*}[ht]
      \centering
\begin{minipage}{0.24\textwidth}
        \centering
\begin{tikzpicture}[scale=0.5, every node/.style={font=\footnotesize}]
\node[smallvertex, label=above:{\(\AAA_2\)}] (v2) at (0,1) {};
          \node[smallvertex, label=left:{\(\,\AAA_1\)}]  (v1) at (-0.9,0) {};
          \node[smallvertex, label=right:{\(\AAA_2\)\,}] (v3) at (0.9,0) {};
          \node[smallvertex, label=below:{\(\AAA_1\)}] (v4) at (0,-1) {};
\foreach \a/\b in {v1/v2,v1/v3,v2/v3,v2/v4,v3/v4}
            \draw (\a) -- (\b);
        \end{tikzpicture}
        \subcaption{ \(F\) }
      \end{minipage}\hfill
      \begin{minipage}{0.24\textwidth}
        \centering
\begin{tikzpicture}[scale=0.5, every node/.style={font=\footnotesize}]
          \node[smallvertex, label=above:{\(\AAA_2\)}] (g2) at (0,1) {};
          \node[smallvertex, label=left:{\(\AAA_1\)}]  (g1) at (-0.9,0) {};
          \node[smallvertex, label=right:{\(\AAA_2\)}] (g3) at (0.9,0) {};
          \draw (g1) -- (g2) -- (g3) -- (g1);
        \end{tikzpicture}
        \subcaption{ \(G\) }
      \end{minipage}\hfill
      \begin{minipage}{0.24\textwidth}
        \centering
\begin{tikzpicture}[scale=0.5, every node/.style={font=\footnotesize}]
\node[smallvertex, label=above:{\(\AAA\)}] (V2) at (0,1) {};
          \node[smallvertex, label=left:{\(\AAA\)}]  (V1) at (-0.9,0) {};
          \node[smallvertex, label=right:{\(\AAA\)}] (V3) at (0.9,0) {};
          \node[smallvertex, label=below:{\(\AAA\)}] (V4) at (0,-1) {};
\foreach \a/\b in {V1/V2,V1/V3,V2/V3,V2/V4,V3/V4}
            \draw (\a) -- (\b);
    
\node[smallvertex, label=above:{\(\SSS\)}] (s) at (-1.35,2.1) {};
          \node[smallvertex, label=above:{\(\XXX\)}] (p1) at (-0.45,2.1) {};
          \node[smallvertex, label=above:{\(\XXX\)}] (p2) at (0.45,2.1) {};
          \node[smallvertex, label=above:{\(\TTT\)}] (t) at (1.35,2.1) {};
          
          \draw (s) -- (p1) -- (p2) -- (t);
    
          \draw (V1) -- (p1);
          \draw (V4) -- (p1);
          \draw (V2) -- (p2);
          \draw (V3) -- (p2);
        \end{tikzpicture}
        \subcaption{ \(F^{*}\) }
      \end{minipage}\hfill
      \begin{minipage}{0.24\textwidth}
        \centering
\begin{tikzpicture}[scale=0.5, every node/.style={font=\footnotesize}]
\node[smallvertex, label=above:{\(\AAA\)}] (G2) at (0,1) {};
          \node[smallvertex, label=left:{\(\AAA\)}]  (G1) at (-0.9,0) {};
          \node[smallvertex, label=right:{\(\AAA\)}] (G3) at (0.9,0) {};
          \draw (G1) -- (G2) -- (G3) -- (G1);
    
\node[smallvertex, label=above:{\(\SSS\)}] (s) at (-1.35,2.1) {};
          \node[smallvertex, label=above:{\(\XXX\)}] (q1) at (-0.45,2.1) {};
          \node[smallvertex, label=above:{\(\XXX\)}] (q2) at (0.45,2.1) {};
          \node[smallvertex, label=above:{\(\TTT\)}] (t) at (1.35,2.1) {};
    
          \draw (s) -- (q1) -- (q2) -- (t);
          
          \draw (G1) -- (q1);
          \draw (G2) -- (q2);
          \draw (G3) -- (q2);
        \end{tikzpicture}
        \subcaption{ \(G^{*}\) }
      \end{minipage}
    
      \caption{
        \(F\) is a \(K_4\) minus one edge; two vertices share colour \(\AAA_1\) so the missing edge is between them (this enables a homomorphism to \(G\)).
        \(G\) is the triangle \(\AAA_1,\AAA_2,\AAA_2\).
        \(F^{*}\) (resp.\ \(G^{*}\)) augments the base graph by a $2$-\emph{colour gadget} \(\SSS\!-\!\XXX\!-\!\XXX\!-\!\TTT\) and, in correspondence to their previous colour, connects the base vertices to the gadget.}
      \label{fig:four-graphs}
    \end{figure*}
    
    \begin{claim}\label{claim:4-colour-encoding}
        For all $m$-colour graphs $F,G$ and $n\in \NN$ holds that $G$ satisfies $(F,n)$ if and only if $G^*$ satisfies $(F^*,n)$.
    \end{claim}
    
    \begin{proof}[Proof of \cref{claim:4-colour-encoding}.]
        In any homomorphism $h: F^*\to G^*$, the $m$-colour gadget of $F^*$ has to be mapped identically to the $m$-colour gadget of $G^*$. If $F^*$ contains any $\AAA$-coloured vertices, because they have exactly one $\XXX$-neighbour, $vs_i \in E(F^*)$ holds if and only if $h(v)s_i \in E(G^*)$ does. By construction this is equivalent to $v \in \AAA_i^F \Leftrightarrow h(v) \in \AAA_i^G$, and thus, it holds that
        $
            \hom(F,G) = \abs{\{h|_{V_F} \mid h\in\text{Hom}(F^*,G^*)\}}=\hom(F^*,G^*)$.
        \end{proof} 
    
    For all $i<j \in [m]$ we define the modified $m$-colour gadget $C_{m,i,j} \coloneqq \smallColourCycle{m}{i}{j}$, where the $\AAA$ vertex is connected with the $i$th and $j$th $\XXX$ vertex.
    
    Let $\cons{F}\coloneqq (F_i,m_i)_{i\in [k]}$, where $F_i$ are coloured with $\AAA_1,\dots,\AAA_m$. The reduction produces the following constraints $\cons{F^*}$ with $4$ colour classes:
    \[
         \cons{F^*}\coloneqq \big ( (F_i^*,m_i)_{i\in [k]} \cup \cons{F}_X\cup (C_{m,i,j},0)_{i<j\in [m]}\big ).
    \]
    To show that $\cons{F}\in \textsc{ColHomRec} \Leftrightarrow \cons{F^*}\in \textsc{ColHomRec}$, first, assume that the $m$-colour graph $G=(V(G),E(G),\AAA_1^G,\dots,\AAA_m^G)$ satisfies $\cons{F}$.
    
    Then $G^*=(V(G^*),E({G^*}),\SSS^{G^*},\TTT^{G^*},\XXX^{G^*},\AAA^{G^*})$, the graph we obtain from $G$ with the same construction that produced the list of constraints, satisfies $(F_i^*,m_i)_{i\in [k]}$ because of \cref{claim:4-colour-encoding}. Since $G^*$ contains $G_m$ as subgraph, and all other vertices are coloured $\AAA$, $G^*$ satisfies $\cons{F}_\XXX$. Furthermore, each vertex $v\in V(G^*)$ has at most one $\XXX$-coloured neighbour, and thus, $G^*$ satisfies $(C_{m,i,j},0)$ for each $i<j\in[m]$.
    
    \begin{claim}\label{claim:m:4-correct}
        If a $4$-colour graph $G$ satisfies $\cons {F^*}$, then there exists an $m$-colour graph $G'$ that satisfies $\cons {F}$. Furthermore, $\cons {F^*}$ can be computed in polynomial time from $\cons {F}$.
    \end{claim}
    \cref{claim:m:4-correct} completes the proof of \cref{lemma:m-to-4}.
\end{proof}

\begin{lemma}\label{lemma:4-to-1}
    There exists a polynomial time reduction from {\upshape \textsc{ColHomRec}}\ with at most $4$ colour classes to \HomR\ (with uncoloured simple graphs).
\end{lemma}

This part of the reduction requires a finite number of homomorphically incomparable graphs, called Kneser graphs, that behave similarly to colours when they are attached to vertices. We modify the constraints $(F,m)$ by attaching indicator gadgets to all vertices of $F$. Because we do not disallow adjacency between vertices of the same colour, we cannot avoid replacing edges with the bidirectional gadget from \cite{BokerHRSS24}. For further details, we refer to the appendix.

\section{\texorpdfstring{$\Sigma_2^p$}{Sigma_2^p}-Completeness of Unary Homomorphism Reconstruction}\label{sec:unary}

In contrast to \HomR, the number of homomorphisms in constraints of \UnHomR\ is polynomially bound by the input size. As we will see, \UnHomR\ is contained in the second level of the polynomial hierarchy.
In this section we will be talking about labelled graphs $F$ or $G$, with a set of labels $\CL\coloneqq\{\ell_1,\dots,\ell_m\}$ for some $m\in\NN$, and some labelling function $\ell:\CL\to V(F)$ that we do not explicitly refer to.

\begin{lemma}
    $\textsc{\upshape UnHomRec} \in \Sigma_2^p$.
\end{lemma}
\begin{proof}
    We can assume that the graph $G$ has at most $\sum_{i=1}^m h_i \abs{V(F_i)}$ vertices by the observation from $\cite{BokerHRSS24}$. An $\textsf{NP}^{\textsf{NP}}$ machine first non-deterministically guesses $G$ and $h_i$ different homomorphisms $h: F_i\to G$. We verify that there are no other homomorphisms by calling the $\NP$ oracle.
\end{proof}

\begin{remark}
    We remark that for constraint graphs from a class of bounded treewidth, the problem \UnHomR\ is in \NP\, since counting homomorphisms is in polynomial time. It remains open, whether there are classes for which this problem is \NP-complete.
\end{remark}

The \NonThreeColouringExtension problem is complete for $\Sigma_2^p$ (see Thm. 11.4 in \cite{AjtaiFS00}). Its containment in $\Sigma_2^p$ is obvious, since it is of the form ``$x\in L$ if and only if $y$ exists such that $f(x,y)\in\nonthreecolouring$'' and we recall that $\nonthreecolouring \in \coNP$.

\dproblem{\NonThreeColouringExtension}
{Graph $F$.}
{Is there a 3-colouring of the degree $1$ vertices of $F$ that cannot be extended to a 3-colouring of $F$?}

\begin{lemma}
    $\NonThreeColouringExtension \leq_p \textsc{\upshape UnHomRec}$.
\end{lemma}
\begin{proof}
    We will use labelled graphs as constraints. Encoding the labels as Kneser graph gadgets (see uncoloured hardness for $\NP^{\SP}$ in \cite{BokerHRSS24Arxiv}, Appendix C), this reduction can be adapted to produce unlabelled graphs instead.

    Let $F$ be an instance of \NonThreeColouringExtension, and fix an order $\{v_1,\dots,v_m\}$ on all degree $1$ vertices of $F$. We define the labelled graph $F'$ by labelling each $v_i$ by a distinct label $\ell_i$. The following constraints are created by the reduction:
    
    \begin{multicols}{2}
		\begin{enumerate}
        \item\label{hcfg} $\hom(F') = 0$,
        \item\label{hck1} $\hom(\smallKone) = 3$,
        \item\label{hck2} $\hom(\smallKtwo) = 6$, and
        \item\label{hcell} $\hom(\tikz[baseline=-3pt]{
            \node[smallvertex, label={right:\small$\ell_i$}] (ell_i) {}
        }\hspace{-3pt}) = 1$ for every $i\in[m]$.
        \end{enumerate}
    \end{multicols}

    The graph $\smallKthree$ is the unique graph that satisfies $\hom(\smallKone) = 3$ and $\hom(\smallKtwo) = 6$, and each label from $\{\ell_1,\dots,\ell_m\}$ has to appear exactly once in any graph $G$ satisfying the constraints. If $G$ satisfies $\hom(F') = 0$, then matching the position of labels in $\smallKthree$ to colours of labelled vertices in $F'$ yields a 3-colouring of the degree $1$ vertices of $F$ that cannot be extended.
    
    Conversely, assume that $\hom(F') > 0$ for all graphs $G$ satisfying these constraints. Then each homomorphism from $F'$ to $G$ can be translated into a $3$-colouring of $F$ extending a 3-colouring of its degree $1$ vertices, and since each 3-colouring of degree $1$ vertices is represented by some position of labels in $\smallKthree$, it follows that $F\not\in\NonThreeColouringExtension$.
\end{proof}

\section{Tractability of Homomorphism Reconstruction for Star Counts}
\label{sec:algorithms}

In this section we will derive an algorithm for \textsc{StarHomRec} to prove Theorem \ref{theo:star} and thereby answer an open question from \cite{BokerHRSS24}.

We will observe that the homomorphism counts $\hom(S,G)$ from stars $S$ into a graph $G$ only depend on the degree sequence of $G$.
Exploiting this, our algorithm first reconstructs a degree sequence consistent with the counts, if it exists.
Then it uses the Havel-Hakimi algorithm \cite{Hakimi1962, Havel1955} to construct a graph $G$ with this degree sequence.

We use dynamic programming to reconstruct the degree sequence.
For that we will derive a recursion of the following form:
Given an instance $(F_1,m_1),\dots,(F_k,m_k)$ of \textsc{StarHomRec}, there exists a graph $G$ with degree sequence $(d_1,\dots,d_n)$ and $\hom(F_i,G)=m_i$ for $i\in[k]$ if and only if there exists a graph $G'$ with degree sequence $(d'_1,\dots,d'_{n'})$ and $\hom(F_i,G)=m'_i$ for $i\in[k]$, where
\begin{itemize}
    \item the counts $m'_1,\dots,m'_k$ only depend on $m_1,\dots,m_k$ and $d_1$,
    \item the counts get smaller in each recursion step ($m'_i\leq m_i$ for $i\in[k]$),
    \item the graph gets smaller in each recursion step ($n<n'$),
    \item $(d_2,\dots,d_n)$ can be computed from $(d'_1,\dots,d'_{n'})$ and $d_1$ in polynomial time.
\end{itemize}

This allows us to recursively guess $d_1$ and compute the counts $m'_1,\dots,m'_k$ from $m_1,\dots,m_k$ and $d_1$.
Then we get up to $d_1$ different but smaller instances of \textsc{StarHomRec} that we have to compute.
Since we store all computed results for \textsc{StarHomRec} in a table, the number of recursive calls is still polynomial.

The exact recursion of our algorithm is slightly different (see Algorithm \ref{alg:dynamic_programming}).
This is, for example, because we must only consider graphic degree sequences.
We address this challenge in Section \ref{sec:graphic_degree_sequences}, after showing a recursive formulation of $\hom(S,G)$ for stars $S$ in Section \ref{sec:recursive_counts}.
In Section \ref{sec:algorithm} we present the final algorithm and prove Theorem \ref{theo:star}.

\subsection{Recursive Formulation of Star Counts}
\label{sec:recursive_counts}

By $S_j$, we denote the star with $j$ leaves. For any star $S_j$, graph $G$, and $v\in G$ we write $\hom(S_j,v)$ for the number of homomorphisms from $S_j$ into $G$ for which the star center node is mapped to $v$.
We use $\sub(S_j,v)$ analogously for subgraphs.

\begin{observation}
	\label{obs:graph_counts_to_vertex_counts}
	For any graph $G$ it holds that
	\begin{align*}
		\hom(S_j,G)&=\sum_{v\in G}\hom(S_j,v)\ \ \ \text{for }j\in\N,\\
		\sub(S_j,G)&=\sum_{v\in G}\sub(S_j,v)\ \ \ \ \text{for }j\in\N\setminus\{1\},\\
		2\cdot\sub(S_1,G)&=\sum_{v\in G}\sub(S_1,v).
	\end{align*}
\end{observation}
Note that $\sub(S_0,G)=\hom(S_0,G)=|V(G)|$ and $2\cdot\sub(S_1,G)=\hom(S_1,G)=2|E(G)|$.

The star counts for one vertex only depend on the vertex' degree.
\begin{observation}\label{obs:vertex_counts}
	For any graph $G$ and vertex $v\in G$ with degree $d$ it holds that
	\begin{align*}
		\hom(S_j,v)=d^j,\ \ \ \ \sub(S_j,v)=\binom{d}{j}.
	\end{align*}
\end{observation}
Overall, a graph's structure beyond its degree sequence is irrelevant for its star counts.
To find the recursion for our algorithm we start with recursive formulations for the vertex based counts.
For subgraphs we use
\begin{align*}
	\s_j(d)\coloneqq
	\begin{cases}
		1 & \text{if }\ j=0,\\
		0 & \text{if }\ d=0,\ j>0,\\
		\s_{j}(d-1)+\s_{j-1}(d-1) & \text{otherwise}.
	\end{cases}
\end{align*}
And for homomorphisms we use
\begin{align*}
	\h_j(d)\coloneqq
	\begin{cases}
		1 & \text{if }\ j=0,\\
		0 & \text{if }\ d=0,\ j>0,\\
		\displaystyle \sum_{i=0}^{j}\binom{j}{i}\,\h_i(d-1) & \text{otherwise}.
	\end{cases}
\end{align*}

\begin{lemma}\label{lem:recursive_counts}
	For any $j\in\N$, graph $G$, and vertex $v\in G$ with degree $d$ it holds that
	\begin{align*}
		\hom(S_j,v)=\h_j(d),\ \ \ \ \sub(S_j,v)=\s_j(d).
	\end{align*}
\end{lemma}
The proof of Lemma \ref{lem:recursive_counts} can be found in the appendix.
In the next section, we will develop a recursion on graphic degree sequences for which we use above recursive formulation.

\subsection{Recursion on Graphic Degree Sequences}
\label{sec:graphic_degree_sequences}
Our reconstruction algorithm has to consider exactly all graphic degree sequences.
In this section we explore a way to recursively shrink graphic degree sequences and recompute the subgraph counts for the shrunken sequence.
The following theorem gives a criterion for a sequence of integers being graphic.

\begin{theorem}[Erdős-Gallai \cite{ErdosGallai1960}]
	\label{thm:erdos_gallai}
A non-increasing sequence $(d_1,\dots,d_n)$ of non-negative integers is graphic if and only if $\sum_{i=1}^n d_i$ is even and for every $k$:
\begin{align}
\sum_{i=1}^k d_i \leq k(k-1) + \sum_{i=k+1}^n \min(d_i,k).\label{eq:erdos_gallai}
\end{align}
\end{theorem}
Intuitively, for every selection of the $k$ highest degree vertices, we need enough vertices to connect them to.
To achieve this, the Havel-Hakimi algorithm \cite{Hakimi1962, Havel1955} (also see \cite{Shahriari2021}) recursively sorts the sequence non-increasingly and then removes the first degree ($d_1$) while decrementing the next $d_1$ degrees.
We use a similar (but slightly different) recursive operation.

\begin{lemma}
    \label{lem:degree_sequence_step}
    A non-increasing sequence $(d_1,\dots,d_n)$ of positive integers with $n\geq2$ and $d_1\leq n-1$ is graphic if and only if the sequence
    $
    (d_2-1,\dots,d_n-1,\underbrace{1,\dots,1}_{\mathclap{n-1-d_1\text{ times}}})
    $, ordered decreasingly, is graphic.
\end{lemma}
The proof of Lemma \ref{lem:degree_sequence_step} can be found in the appendix.
We can relate the star counts of both sequences (before and after our operation) in Lemma \ref{lem:recursive_counts_for_entire_graph_with_slack} without knowing anything about the sequences besides $d_1$.
\begin{lemma}
\label{lem:recursive_counts_for_entire_graph_with_slack}
    Let $(d_1,\dots,d_n)$ be a non-increasing sequence of positive integers with $n\geq2$.
	Furthermore, let $G$ be a graph with the degree sequence $(d_1,\dots,d_n)$, and let $G'$ be a graph with the degree sequence
	$
	(d_2-1,\dots,d_n-1,\underbrace{1,\dots,1}_{x\text{ times}}),
	$
    ordered decreasingly, for any $x\in \N$.
    Then
	\begin{align*}
		\sub(S_0,G)
		&=\s_0(d_1)+\sub\left(S_0,G'\right)-x,\\
		2\cdot\sub(S_1,G)
		&=\s_1(d_1)+2\cdot\sub\left(S_1,G'\right)+\sub\left(S_0,G'\right)-2x,\\
		\sub(S_2,G)
		&=\s_2(d_1)+\sub\left(S_2,G'\right)+2\cdot\sub\left(S_1,G'\right)-x,\\
		\sub(S_j,G)
		&=\s_j(d_1)+\sub\left(S_j,G'\right)+\sub\left(S_{j-1},G'\right)\ \ \text{for $j\geq3$}.
	\end{align*}
\end{lemma}
The proof of Lemma \ref{lem:recursive_counts_for_entire_graph_with_slack} can be found in the appendix. It is based on the recursion from Section \ref{sec:recursive_counts} and the factor $2$ for $\sub(S_1,G)$ comes from Observation \ref{obs:graph_counts_to_vertex_counts}.
An analogous version of Lemma \ref{lem:recursive_counts_for_entire_graph_with_slack} for homomorphisms could be made using the recursion for homomorphisms from Section \ref{sec:recursive_counts} and leaving out the factor 2.
However, for simplicity, we only consider subgraph counts in the remainder of Section \ref{sec:algorithms}. As observed in the introduction, for stars, subgraph counts determine homomorphism counts and vice versa, so this is sufficient.

It seems like this operation does not help us for our algorithm because the sequence gets longer.
However, since the last $x$ entries of the sequence all have value 1 (before reordering), we can separate this part from the (still unknown) degree sequence and carry the number $x$ from one recursion step to another, updating it each step.

Intuitively, if we decrement more degrees than necessary to ensure that the degree $d_1$ can be realized, then, when another degree would not be realisable using only the rest of the degree sequence, we can use up $x$.

\subsection{The Reconstruction Algorithm Using Dynamic Programming}
\label{sec:algorithm}
In this section we first lift our recursion on graphic degree sequences to a recursion on the reconstruction problem.
Then we present the reconstruction algorithm which is based on this recursion and discuss its runtime to finally prove Theorem \ref{theo:star}.

\begin{definition}
	\label{def:recursive_solve}
	We define the function $\dpf:\N^{\ell+3}\to\{0,1\}$ with $\dpf(s_0,\dots,s_\ell,m,x)=1$ if there exist integers $d_1\ge\dots\ge d_{s_0}\geq0$ and a graph $G$ with degree sequence
    \[
	(d_1,\dots,d_{s_0},\underbrace{1,\dots,1}_{x\text{ times}}),
	\]
	ordered decreasingly, which fulfils
	\begin{align*}
		&\sub(S_0,G)
		=\s_0+x,\\
		2\cdot&\sub(S_1,G)
		=\s_1+x,\\
		&\sub(S_j,G)
		=\s_j\ \ \ \ \ \text{for $j\in[2,\ell]$},\\
		&d_i\leq m\ \ \ \ \ \text{for $i\in[1,s_0]$},
	\end{align*}
    and $\dpf(s_0,\dots,s_\ell,m,x)=0$ otherwise.
    
\end{definition}
Note that $\dpf(s_0,s_1/2,s_3,\dots,s_\ell,s_0-1,0)=1$ if and only if there exists a graph $G$ with $\sub(S_j,G)=s_j$ for all $j\in[0,\ell]$.
Intuitively, the parameter $x$ corresponds to the $x$ that we discussed in Section \ref{sec:graphic_degree_sequences} and $m$ acts as a maximum allowed degree, ensuring the (degree) sequence to be non-increasing which is a requirement for Lemma \ref{lem:degree_sequence_step} and Lemma \ref{lem:recursive_counts_for_entire_graph_with_slack}.

Using this definition we can solve the problem recursively up to the ``base cases'' with $s_0<2$ or $s_1=0$.

\begin{lemma}
\label{lem:problem_recursion}
For $\ell\geq1$ and all non-negative integers $s_0,\dots,s_\ell,m,x$ with $s_0\geq2$ and $s_1\geq1$ it holds that $\dpf(s_0,\dots,s_\ell,m,x)$ is equal to
\begin{align*}
\max(\{\dpf(s'_0,\dots,s'_\ell,m',x')\ |\ 
&\exists d \in [1,m],\ 
\exists n_1 \in [1,s_0]:\\
&d\leq n_1-1+x,\\
&s'_0=n_1-1,\\
&s'_j=s_j-s_j(d)-s'_{j-1}\ \text{ for $j\in[1,\ell]$},\\
&m'=d-1,\\
&x'=x+n_1-1-d
\}).
\end{align*}

\end{lemma}
In the proof of Lemma \ref{lem:problem_recursion}, which can be found in the appendix, we first prune entries of the degree sequence with value $0$ and then apply Lemma \ref{lem:degree_sequence_step} together with Lemma \ref{lem:recursive_counts_for_entire_graph_with_slack}.
This pruning step is implemented implicitly by guessing the number $n_1\in[1,s_0]$ of vertices with degree at least 1.

We now combine Lemma \ref{lem:problem_recursion} and the Havel-Hakimi algorithm \cite{Hakimi1962, Havel1955} into a reconstruction algorithm for $\textsc{StarSubRec}$ where we assume the constraint graphs $F_1,\dots,F_k$ to be exactly all stars of size up to $k$.

\begin{algorithm}[ht]
\caption{DP$(s_0,\dots,s_\ell, m, x)$}
\label{alg:dynamic_programming}
\DontPrintSemicolon

\If(\tcp*[f]{Base cases (part 1)}){$s_1 = 0$}{
  \Return{$(s_2=0 \land \dots \land s_\ell=0)$}\tcp*[f]{Empty graph or isolated nodes}
}

\If(\tcp*[f]{Base cases (part 2)}){$s_0 = 1$}{
  \Return{$(s_1=0 \land \dots \land s_\ell=0) \lor (s_1=1 \land s_2=0 \land \dots \land s_\ell=0 \land m\ge 1 \land x\ge 1)$}
}

\For(\tcp*[f]{Guess highest degree}){$d \gets 1$ \KwTo $m$}{
  \For(\tcp*[f]{Guess number of vertices with non-zero degree}){$n_1 \gets 1$ \KwTo $s_0$}{
    \If(\tcp*[f]{Check realizability of highest degree}){$d \le (n_1-1)+x$}{
      $s'_0 \gets n_1 - 1$\;
      $s'_1 \gets s_1 - \binom{d}{1} - s'_0$\;
      $\dots$\;
      $s'_\ell \gets s_\ell - \binom{d}{\ell} - s'_{\ell-1}$\;
      $m' \gets d-1$\;
      $x' \gets x + n_1 - 1 - d$\;
      \If(\tcp*[f]{Recursion}){$\textsc{\upshape DP}(s'_0,\dots,s'_\ell, m', x')$}{
        \Return{\textbf{True}}
      }
    }
  }
}

\Return{\textbf{False}}
\end{algorithm}

\begin{theorem}
    \label{theo:stars_all_specified}
    The problem $\textsc{\upshape StarSubRec}$ can be solved in time $n^{O(k^2)}$ if $F_i=S_{i-1}$ for all $i\in[1,k]$ where $n$ is the count from the constraint $(S_0,n)$.
    \begin{proof}
	   We give an algorithm $\mathcal{A}$ for $\textsc{StarSubRec}$.
       Let $(F_1,m_1),\dots,(F_k,m_k)$ be the input of $\textsc{StarSubRec}$.
       We denote this input as $(S_0,s_0),\dots,(S_\ell,s_\ell)$.
    Let $n\coloneqq s_0$.
	First, $\mathcal{A}$ computes $\dpf(s_0,s_1/2,s_2,\dots,s_\ell,n-1,0)$ using dynamic programming with the recursion from Lemma \ref{lem:problem_recursion} and $\ell+3$ table dimensions $s_0,\dots,s_\ell,m,x$.
	The dynamic program always runs into a base case ($s_0<2$ or $s_1<1$) which is trivial to solve (see Algorithm \ref{alg:dynamic_programming}) and therefore terminates because the parameters $s_0,\dots,s_\ell,m$ all get smaller at least by $1$ in each recursive call.
	This happens in runtime $n^{O(\ell^2)}$ because all parameters $s_0,\dots,s_\ell,m,x$ are bounded by $n^\ell$ so there are at most $n^{O(\ell^2)}$ entries in the dynamic programming table that $\mathcal{A}$ might compute and in each recursion step $\mathcal{A}$ tries at most $n^2$ combinations of the values $d$ and $n_1$.
	Algorithm \ref{alg:dynamic_programming} shows pseudo-code for the computation of $\dpf(s_0,\dots,s_\ell,m,x)$.

	If $\dpf(s_0,\dots,s_\ell,n-1,0)=0$ then $\mathcal{A}$ outputs False.
	Otherwise $\mathcal{A}$ reconstructs the degree sequence $(d_1,\dots,d_n)$ from the dynamic programming using backtracking.
	Specifically, the parameter $d$ in the first recursive call from Lemma \ref{lem:problem_recursion} corresponds to $d_1$.
	Similarly, the parameter $d$ from the recursive call number $i$ is $d_i$.
	Finally, $\mathcal{A}$ runs the Havel-Hakimi algorithm \cite{Hakimi1962,Havel1955} (also see \cite{Shahriari2021}) to reconstruct a graph $G$ with the degree sequence $(d_1,\dots,d_n)$ in $O(n^2)$ time and then outputs $G$.
    \end{proof}
\end{theorem}

As already mentioned we get the same result for homomorphisms if we use the recursion $h_j(d)$ from Lemma \ref{lem:recursive_counts}.

To prove Theorem \ref{theo:star} we have to deal with unspecified counts for stars $S_j$ with $j<\ell$.
Since they are all bounded by the largest homomorphism count, we can simply guess the ``missing'' counts and call $\mathcal{A}$ for each guess.

\begin{proof}[Proof of Theorem \ref{theo:star}]

	  We give an algorithm $\mathcal{B}$ for \textsc{StarHomRec}.
    Let $(F_1,m_1),\dots,(F_k,m_k)$ be an instance of \textsc{StarHomRec}.
    Let $\ell\coloneqq\max_{i\in[k]}|F_i|-1$ and $m\coloneqq\{m_i\mid i\in[k]\}$.
    Each $F_i$ is a star $S_j$ with $j\leq\ell$ but there might be stars $S_j$ with $j<\ell$ for which there is no count specified.
    $\mathcal{B}$ guesses all ``missing'' star counts up to the star $S_\ell$ so that we get constraints $(S_0,s_0),\dots,(S_\ell,s_\ell)$ where $s_j=m_i$ for all $i\in[0,k],j\in[0,\ell]$ with $S_j=F_i$ .
    For each guess, $\mathcal{B}$ calls the homomorphism variant of algorithm $\mathcal{A}$ from the proof of Theorem \ref{theo:stars_all_specified} with input $(S_0,s_0),\dots,(S_\ell,s_\ell)$.

    $\mathcal{B}$ calls $\mathcal{A}$ at most $m^\ell$ times because it has to guess at most $\ell$ counts which are all bounded by $s_1\leq s_2\leq\dots\leq s_\ell\leq m$ and we can also bound $s_0\leq s_1$ because $s_0>s_1$ would just correspond to isolated vertices which are trivial to add.
    Overall, $\mathcal{B}$ has a runtime of $m^\ell m^{O(k^2)}=m^{O(k^2)}$.
    
\end{proof}

\section{Conclusions}
We determine the exact computational complexity of the homomorphism
reconstruction problem both in the case that the constraints are
specified in binary and unary: the (binary) \textsc{HomRec} problem is
\textsf{NEXPTIME}-complete, and the (unary) \textsc{UnHomRec} is
$\Sigma_2^p$-complete.

Consider the following parameterized
version of the unary reconstruction problem. It was proved in \cite{BokerHRSS24} that the binary version of this parameterized problem is hard.

\pproblem{$p\textsc{-UnHomRec}$}
  {Pairs $(F_1, m_1), \dots, (F_k,m_k)$, where $F_1,\ldots,F_k$
    are graphs and $m_1,\ldots,m_\ell\in\mathbb N$ (in unary encoding).}{$\sum_{i=1}^k|F_i|$}
  {Is there a graph $G$ such that
    $\hom(F_i, G) = m_i$ for every $i \in [k]$?}

  It is open if this problem is fixed-parameter tractable, or at least
  the complexity class XP, that is, solvable in polynomial
  time for each fixed parameter value. 
We prove that the problem is in
XP if the $F_i$ are stars.

\bibliography{literature}

@article{AjtaiFS00,
  author       = {Miklos Ajtai and
                  Ronald Fagin and
                  Larry J. Stockmeyer},
  title        = {The Closure of Monadic {NP}},
  journal      = {J. Comput. Syst. Sci.},
  volume       = {60},
  number       = {3},
  pages        = {660--716},
  year         = {2000},
  doi          = {10.1006/jcss.1999.1691}
}

@article{GalperinW83,
  author       = {Hana Galperin and
                  Avi Wigderson},
  title        = {Succinct Representations of Graphs},
  journal      = {Inf. Control.},
  volume       = {56},
  number       = {3},
  pages        = {183--198},
  year         = {1983},
  doi          = {10.1016/S0019-9958(83)80004-7}
}

@article{PapadimitriouY86,
	author = {Christos H. Papadimitriou and Mihalis Yannakakis},
	journal = {Information and Control},
	number = {3},
	pages = {181-185},
	title = {A note on succinct representations of graphs},
	volume = {71},
	year = {1986},
    doi = {10.1016/S0019-9958(86)80009-2}
}

@inproceedings{JinBCL24,
	author = {Emily Jin and Michael M. Bronstein and {\.I}smail {\.I}lkan Ceylan and Matthias Lanzinger},
	bibsource = {dblp computer science bibliography, https://dblp.org},
	booktitle = {Forty-first International Conference on Machine Learning, {ICML} 2024, Vienna, Austria, July 21-27, 2024},
	publisher = {OpenReview.net},
	title = {Homomorphism Counts for Graph Neural Networks: All About That Basis},
	url = {https://openreview.net/forum?id=zRrzSLwNHQ},
	year = {2024}}

@article{WolfOPG23a,
	author = {Hinrikus Wolf and Luca Oeljeklaus and Pascal K{\"u}hner and Martin Grohe},
	doi = {10.48550/arXiv.2308.15283},
	journal = {ArXiv},
	title = {Structural Node Embeddings with Homomorphism Counts},
	volume = {arXiv:2308.15283},
	year = {2023}}

@article{GroheRS25,
	author = {Martin Grohe and Gaurav Rattan and Tim Seppelt},
	doi = {10.19086/aic.2025.4},
	journal = {Advances in Combinatorics},
    volume = {2025},
	number = {4},
	status = {JOU},
	title = {Homomorphism Tensors and Linear Equations},
	year = {2025}}

@inproceedings{BokerHRSS24,
	author = {B\"{o}ker, Jan and H\"{a}rtel, Louis and Runde, Nina and Seppelt, Tim and Standke, Christoph},
	booktitle = {Proceedings of the 41st International Symposium on Theoretical Aspects of Computer Science},
	doi = {10.4230/LIPIcs.STACS.2024.19},
	editor = {Beyersdorff, Olaf and Kant\'{e}, Mamadou Moustapha and Kupferman, Orna and Lokshtanov, Daniel},
	pages = {19:1--19:20},
	publisher = {Schloss Dagstuhl -- Leibniz-Zentrum f{\"u}r Informatik},
	series = {LIPIcs},
	title = {{The Complexity of Homomorphism Reconstructibility}},
	volume = {289},
	year = {2024},
}

@misc{BokerHRSS24Arxiv,
      title={The Complexity of Homomorphism Reconstructibility}, 
      author={Jan Böker and Louis Härtel and Nina Runde and Tim Seppelt and Christoph Standke},
      year={2023},
      eprint={2310.09009},
      archivePrefix={arXiv},
      primaryClass={cs.DS},
      url={https://arxiv.org/abs/2310.09009}, 
}

@article{razborov_minimal_2008,
	author = {Razborov, Alexander A.},
	doi = {10.1017/S0963548308009085},
	issn = {0963-5483, 1469-2163},
	journal = {Combinatorics, Probability and Computing},
	language = {en},
	month = jul,
	number = {4},
	pages = {603--618},
	title = {On the {Minimal} {Density} of {Triangles} in {Graphs}},
	urldate = {2020-04-28},
	volume = {17},
	year = {2008}}

@article{oneil_ulam_1970,
	author = {P. V. {O'Neil}},
	doi = {10.2307/2316851},
	issn = {00029890, 19300972},
	journal = {The American Mathematical Monthly},
	number = {1},
	pages = {35--43},
	publisher = {Mathematical Association of America},
	title = {Ulam's Conjecture and Graph Reconstructions},
	urldate = {2023-08-17},
	volume = {77},
	year = {1970}}

@article{seppelt_logical_2023,
	author = {Tim Seppelt},
	doi = {10.1016/j.ic.2024.105224},
	journal = {Information and Computation},
	pages = {105224},
	title = {Logical equivalences, homomorphism indistinguishability, and forbidden minors},
	volume = {301},
	year = {2024}}

@article{DellGR18,
	author = {Dell, Holger and Grohe, Martin and Rattan, Gaurav},
	collaborator = {Wagner, Michael},
	copyright = {Creative Commons Attribution 3.0 Unported license (CC-BY 3.0)},
	doi = {10.4230/LIPICS.ICALP.2018.40},
	journal = {45th International Colloquium on Automata, Languages, and Programming (ICALP 2018)},
	language = {en},
	pages = {40:1--40:14},
	title = {Lov{\'a}sz {Meets} {Weisfeiler} and {Leman}},
	urldate = {2020-04-29},
	year = {2018}}

@inproceedings{chaudhuri_optimization_1993,
	author = {Chaudhuri, Surajit and Vardi, Moshe Y.},
	booktitle = {Proceedings of the {Twelfth} {ACM} {SIGACT}-{SIGMOD}-{SIGART} {Symposium} on {Principles} of {Database} {Systems}, {May} 25-28, 1993, {Washington}, {DC}, {USA}},
	doi = {10.1145/153850.153856},
	editor = {Beeri, Catriel},
	pages = {59--70},
	publisher = {ACM Press},
	title = {Optimization of {Real} {Conjunctive} {Queries}},
	url = {https://doi.org/10.1145/153850.153856},
	year = {1993}}

@inproceedings{MancinskaR20,
	author = {Man{\v c}inska, Laura and Roberson, David E.},
	booktitle = {2020 {IEEE} 61st {Annual} {Symposium} on {Foundations} of {Computer} {Science} ({FOCS})},
	doi = {10.1109/FOCS46700.2020.00067},
	pages = {661--672},
	title = {Quantum isomorphism is equivalent to equality of homomorphism counts from planar graphs},
	year = {2020}}

@article{ko_three_1991,
	author = {Ko, Ker-I and Tzeng, Wen-Guey},
	doi = {10.1007/BF01200064},
	issn = {1420-8954},
	journal = {computational complexity},
	month = sep,
	number = {3},
	pages = {269--310},
	title = {Three {$\Sigma_2^P$}-complete problems in computational learning theory},
	url = {https://doi.org/10.1007/BF01200064},
	volume = {1},
	year = {1991}}

@inproceedings{CurticapeanDM17,
	author = {Curticapean, Radu and Dell, Holger and Marx, D{\'a}niel},
	booktitle = {Proceedings of the 49th Annual {ACM} {SIGACT} Symposium on Theory of Computing},
	doi = {10.1145/3055399.3055502},
	isbn = {978-1-4503-4528-6},
	location = {New York, {NY}, {USA}},
	note = {event-place: Montreal, Canada},
	pages = {210--223},
	publisher = {Association for Computing Machinery},
	series = {{STOC} 2017},
	title = {{Homomorphisms Are a Good Basis for Counting Small Subgraphs}},
	url = {https://doi.org/10.1145/3055399.3055502},
	year = {2017}}

@article{kopparty_homomorphism_2010,
	author = {Swastik Kopparty and Benjamin Rossman},
	doi = {10.1016/j.ejc.2011.03.009},
	issn = {0195-6698},
	journal = {European Journal of Combinatorics},
	note = {Homomorphisms and Limits},
	number = {7},
	pages = {1097-1114},
	title = {The homomorphism domination exponent},
	volume = {32},
	year = {2011}}

@article{hatami_undecidability_2011,
	author = {Hatami, Hamed and Norine, Serguei},
	doi = {10.1090/S0894-0347-2010-00687-X},
	issn = {0894-0347},
	journal = {Journal of the American Mathematical Society},
	language = {en},
	month = may,
	number = {2},
	pages = {547--547},
	title = {Undecidability of linear inequalities in graph homomorphism densities},
	urldate = {2020-04-28},
	volume = {24},
	year = {2011}}

@book{Lovasz12,
	address = {Providence, Rhode Island},
	author = {Lov{\'a}sz, L{\'a}szl{\'o}},
	doi = {10.1090/coll/060},
	isbn = {978-0-8218-9085-1},
	number = {volume 60},
	publisher = {American Mathematical Society},
	series = {American {Mathematical} {Society} colloquium publications},
	title = {Large networks and graph limits},
	year = {2012}}

@article{Lovasz67,
	author = {Lov{\'a}sz, L{\'a}zl{\'o}},
	doi = {10.1007/BF02280291},
	issn = {1588-2632},
	journal = {Acta Mathematica Academiae Scientiarum Hungarica},
	month = sep,
	number = {3},
	pages = {321--328},
	title = {Operations with structures},
	url = {https://doi.org/10.1007/BF02280291},
	volume = {18},
	year = {1967}}

@inproceedings{Grohe20c,
	author = {Martin Grohe},
	bibsource = {dblp computer science bibliography, https://dblp.org},
	booktitle = {Proceedings of the 39th {ACM} {SIGMOD-SIGACT-SIGAI} Symposium on Principles of Database Systems, {PODS} 2020, Portland, OR, USA, June 14-19, 2020},
	doi = {10.1145/3375395.3387641},
	editor = {Dan Suciu and Yufei Tao and Zhewei Wei},
	pages = {1--16},
	publisher = {{ACM}},
	title = {word2vec, node2vec, graph2vec, X2vec: Towards a Theory of Vector Embeddings of Structured Data},
	url = {https://doi.org/10.1145/3375395.3387641},
	year = {2020}}

@article{Dvorak10,
	author = {Dvo{\v r}{\'a}k, Zden{\v e}k},
	doi = {10.1002/jgt.20461},
	issn = {03649024},
	journal = {Journal of Graph Theory},
	language = {en},
	month = aug,
	number = {4},
	pages = {330--342},
	shorttitle = {On recognizing graphs by numbers of homomorphisms},
	title = {On recognizing graphs by numbers of homomorphisms},
	urldate = {2020-11-27},
	volume = {64},
	year = {2010}}

@inproceedings{boker_complexity_2019,
	address = {{Dagstuhl, Germany}},
	author = {B{\"o}ker, Jan and Chen, Yijia and Grohe, Martin and Rattan, Gaurav},
	booktitle = {44th {{International Symposium}} on {{Mathematical Foundations}} of {{Computer Science}} ({{MFCS}} 2019)},
	doi = {10.4230/LIPIcs.MFCS.2019.54},
	editor = {Rossmanith, Peter and Heggernes, Pinar and Katoen, Joost-Pieter},
	isbn = {978-3-95977-117-7},
	issn = {1868-8969},
	pages = {54:1--54:13},
	publisher = {{Schloss Dagstuhl\textendash Leibniz-Zentrum fuer Informatik}},
	series = {Leibniz {{International Proceedings}} in {{Informatics}} ({{LIPIcs}})},
	title = {The {{Complexity}} of {{Homomorphism Indistinguishability}}},
	volume = {138},
	year = {2019}}

@inproceedings{Grohe20b,
	address = {New York, NY, USA},
	author = {Grohe, Martin},
	booktitle = {Proceedings of the 35th {Annual} {ACM}/{IEEE} {Symposium} on {Logic} in {Computer} {Science}},
	doi = {10.1145/3373718.3394739},
	isbn = {978-1-4503-7104-9},
	pages = {507--520},
	publisher = {Association for Computing Machinery},
	series = {{LICS} '20},
	title = {Counting {Bounded} {Tree} {Depth} {Homomorphisms}},
	year = {2020}}

@book{graham1994concrete,
	address = {Reading, MA, USA},
	author = {Ronald L. Graham and Donald E. Knuth and Oren Patashnik},
	edition = {2nd},
	isbn = {978-0201558029},
	publisher = {Addison-Wesley},
	title = {Concrete Mathematics: A Foundation for Computer Science},
	year = {1994}}

@article{ErdosGallai1960,
	author = {Paul Erd{\H{o}}s and Tibor Gallai},
	journal = {Matematikai Lapok},
	pages = {264--274},
	title = {Gr{\'a}fok el{\H o}{\'\i}rt foksz{\'a}m{\'u} pontokkal (Graphs with prescribed degrees of vertices)},
	volume = {11},
	year = {1960}}

@article{Havel1955,
	author = {V{\'a}clav Havel},
	journal = {\v{C}asopis pro p\v{e}stov{\'a}n{\'\i} matematiky},
	number = {4},
	pages = {477--480},
	title = {Pozn{\'a}mka o existenci kone\v{c}n{\'y}ch graf{\r{u}} (A remark on the existence of finite graphs)},
	url = {https://eudml.org/doc/19050},
	volume = {80},
	year = {1955}}

@article{Hakimi1962,
	author = {S. L. Hakimi},
	doi = {10.1137/0110037},
	journal = {Journal of the Society for Industrial and Applied Mathematics},
	number = {3},
	pages = {496--506},
	title = {On the Realizability of a Set of Integers as Degrees of the Vertices of a Linear Graph},
	volume = {10},
	year = {1962}}

@book{Shahriari2021,
  title={An invitation to combinatorics},
  author={Shahriari, Shahriar},
  year={2021},
  publisher={Cambridge University Press}
}

@article{Tinhofer91,
	author = {Gottfried Tinhofer},
	journal = {Discrete Applied Mathematics},
	number = {2-3},
	pages = {253--264},
	title = {A note on compact graphs},
	volume = {30},
	year = {1991},
    doi = {10.1016/0166-218X(91)90049-3}
}

@inproceedings{Seppelt24,
  author       = {Tim Seppelt},
  editor       = {Rastislav Kr{\'{a}}lovic and
                  Anton{\'{\i}}n Kucera},
  title        = {An Algorithmic Meta Theorem for Homomorphism Indistinguishability},
  booktitle    = {49th International Symposium on Mathematical Foundations of Computer
                  Science, {MFCS} 2024, Bratislava, Slovakia, August 26-30, 2024},
  series       = {LIPIcs},
  volume       = {306},
  pages        = {82:1--82:19},
  publisher    = {Schloss Dagstuhl - Leibniz-Zentrum f{\"{u}}r Informatik},
  year         = {2024},
  url          = {https://doi.org/10.4230/LIPIcs.MFCS.2024.82},
  doi          = {10.4230/LIPICS.MFCS.2024.82},
}

@inproceedings{KarRS025,
	author = {Prem Nigam Kar and David E. Roberson and Tim Seppelt and Peter Zeman},
	booktitle = {52nd International Colloquium on Automata, Languages, and Programming, {ICALP} 2025, Aarhus, Denmark, July 8-11, 2025},
	doi = {10.4230/LIPICS.ICALP.2025.105},
	editor = {Keren Censor{-}Hillel and Fabrizio Grandoni and Jo{\"{e}}l Ouaknine and Gabriele Puppis},
	pages = {105:1--105:19},
	publisher = {Schloss Dagstuhl - Leibniz-Zentrum f{\"{u}}r Informatik},
	series = {LIPIcs},
	title = {{NPA} Hierarchy for Quantum Isomorphism and Homomorphism Indistinguishability},
	volume = {334},
	year = {2025},
}

@incollection{CernyS26,
	author = {Marek {\v C}ern{\'y} and Tim Seppelt},
	booktitle = {Proceedings of the 43rd International Symposium on Theoretical Aspects of Computer Science (this conference)},
	publisher = {Schloss Dagstuhl - Leibniz-Zentrum f{\"{u}}r Informatik},
	series = {LIPIcs},
	title = {Homomorphism Indistinguishability, Multiplicity Automata Equivalence, and Polynomial Identity Testing},
	year = {2026}
}

@article{MarcinkowskiO25,
	author = {Jerzy Marcinkowski and Piotr Ostropolski-Nalewaja},
	journal = {ArXiv},
	title = {Bag Semantics Query Containment: The CQ vs. UCQ Case and Other Stories},
	volume = {2503.07219},
	year = {2025}}

@article{MarcinkowskiO24,
	author = {Marcinkowski, Jerzy and Orda, Mateusz},
	journal = {Proceedings of the ACM on Management of Data},
	number = {2},
	pages = {1--24},
	title = {Bag Semantics Conjunctive Query Containment. Four Small Steps Towards Undecidability.},
	volume = {2},
	year = {2024}}

@inproceedings{ChandraM77,
	author = {A.K. Chandra and P.M. Merlin},
	booktitle = {Proceedings of the 9th ACM Symposium on Theory of Computing},
	pages = {77-90},
	title = {Optimal Implementation of Conjunctive Queries in Relational Data Bases},
	year = {1977}
}

@article{Kelly57,
	author = {Paul Joseph Kelly},
	journal = {Pacific Journal of Mathematics},
	pages = {961-968},
	title = {A congruence theorem for trees},
	volume = {7},
	year = {1957}
}

@book{Ulam60,
	author = {Stanis{\l}aw Ulam},
	publisher = {Wiley},
	title = {A collection of mathematical problems},
	year = {1960}}

@inproceedings{Babai95,
	author = {L{\'{a}}szl{\'{o}} Babai},
	booktitle = {Handbook of Combinatorics},
	editor = {R. Graham and M. Gr{\"o}tschel and L. Lov{\'a}sz},
	pages = {1447--1540},
	publisher = {MIT Press},
	title = {Automorphism groups, isomorphism, reconstruction},
	volume = {2},
	year = {1995}
}

\newpage
\appendix

\section{Missing Proofs from Section \ref{sec:binary}}
    \begin{proof}[Proof of \cref{claim:3.1-correct}.]
    Recall that
    \[ \cons{F} \coloneqq \cons{F_\alpha} \cupp \big ((\smallPtwoOut, 1)\big ) \cupp \bigcup_{v \in V(C)} \cons{F_{\val \cdot}}(\Cir_v, 1) \cupp \bigcup_{v \in V(C)\setminus \Inp^C} \cons{F}_C(v) .\]
        \begin{itemize}
            \item By our construction, any graph $G$ satisfying $\cons{F_\alpha}$ contains exactly one copy of the value gadget $\smallAlpha$ as subgraph. Furthermore, no other vertices in $G$ are coloured by either $\alpha$, $\bot$, or $\top$ (because their total counts in $G$ are also 2, 1, and 1).
\item If $G$ satisfies $\cons{F_{\val \cdot}}(\Cir_v, 1)$ for each node $v\in V(C)$, we observe that $G$ contains exactly one node $v'$ coloured $\Cir_v$ per node $v\in V(C)$. There is a single edge connecting $v'$ with either one of the two vertices coloured $\alpha$. With our observation from earlier, it follows that $\val{v'}$ is defined in $G$, and the string $\val{v'_1}\dots\val{v'_n}\in\{0,1\}^n$ of values of input nodes $v'_{1},\dots,v'_{n}$ in $G$ defines an input $x$ of $C$.
            \item Should $G$ also satisfy $\cons{F}_C(v)$ for each gate $v \in V(C)\setminus \Inp^C$, then we claim that $G$ encodes an evaluation of $C(x)$. 
            We will proof the inductive case of $v$ being an $\land$-gate $\land(u,w)$, for incoming edges from nodes $u,w\in V(C)$. The cases for $\lor$-gate and $\neg$-gates are similar. Let $v',u',w' \in V(G)$ be the vertices coloured $\Cir_v,\Cir_u,\Cir_w$, respectively. Fix some input $x$, as described in the previous bullet. By our inductive hypothesis $\val{u'}$ and $\val{w'}$ are already computed correctly in the represented evaluation of $C(x)$ in $G$ and, say, it holds that $\val{u'}=0$ and $\val{w'}=1$ (other cases for $\val{u'},\val{w'}\in\{0,1\}^2$ correspond to another constraint in $\cons{F}_C(v)$ each). Because $G$ satisfies $\cons{F_{\val \cdot}}(\Cir_v, 1)$, $\val{v'}$ is defined and either 0 or 1. Towards a contradiction, assume that $\val{v'} = 1 \neq 0 = \val{u'} \cdot \val{w'}$. Then by definition of $\val{\cdot}$, the vertex $v'$ of $G$ is adjacent to the $\alpha$-coloured next to the $\top$-coloured vertex. Taking into account the values of $u'$ and $w'$, $G$ must contain $\smallAlphaAnd{0}{1}{1}$ as subgraph. This is a direct contradiction to the assumption that $G$ satisfies constraint $(\smallAlphaAnd{0}{1}{1},0)$ in $\cons{F}_C(v)$.
            \item Finally, we observe that $G$ satisfying $(\smallPtwoOut, 1)$ is equivalent to $\val{o'}=1$ for the $\Cir_o$ vertex $o'$ corresponding to output node $o$, and thus to $C(x)=1$, for input $x = \val{v'_{1}}\dots\val{v'_{n}}$. 
        \end{itemize}
    
        We have computed $O(\abs{C})$ constraints that consist of graphs of at most 7 vertices, with homomorphism counts either 0, 1, or 2. In total, the graphs contain $\abs{C}+3$ different colours, which can be encoded in binary strings of length $O(\log\abs{C})$ each. This brings the instance produced by our reduction to $O(\abs{C}\log\abs{C})$ size (since we got around using anything but $O(1)$ homomorphism counts, this bound remains even if encoded in unary). Obviously, each individual constraint can be computed in polynomial time from the input circuit $C$.    
    \end{proof}

    \begin{proof}[Proof of \cref{claim:3.2-correct}.]
    Recall that
    \[ \cons{F} \coloneqq \cons{F_\alpha} \cupp \big ((\smallPtwoOut, k(k-1)/2)\big ) \cupp \bigcup_{c \in C_G} \big ( \Fnm(c, \alpha, k(k-1)/2, 1) \cupp \cons{F}_c \big ) \cupp F^+. \]
        \begin{itemize}
        \item Analogously to the proof of \cref{claim:3.1-correct}, we begin by observing that any graph $G$ satisfying $\cons{F_\alpha}$ contains exactly one copy of the value gadget $\smallAlpha$ as subgraph and no other vertices in $G$ are coloured by either $\alpha$, $\bot$, or $\top$.
        \item Since $G$ satisfies $\Fnm(\Cir_v, \alpha, k(k-1)/2, 1)$, for each node $v\in V(C)$, $G$ contains $k(k-1)/2$ vertices in colour class $\Cir_v$. Each vertex $v\in\Cir_v$ has a single neighbour in the value gadget, and thus, $\val{v}$ is defined.
        \item From $\cons{F}_\PPP$ then follows that $G$ contains as subgraphs $k(k-1)/2$ stars with a central $\PPP$-coloured vertex and leafs coloured by each $\Cir_v$, for $v\in V(C)$. We refer to them as $\PPP$-stars.
        \item If $G$ also satisfies $\cons{F}_{C_G}(v)$ for each gate $v$, then we claim that $G$ encodes $k(k-1)/2$ distinct evaluations of $C_G$. The proof is almost analogous to the $\CircuitSAT$ case, but, recall that we changed the constraints $\cons{F}_{C_G}(v)$ slightly. By connecting the $\PPP$-coloured vertex to all gate-coloured vertices inside of the constraint graphs, we made the constraints less restrictive. Now, $\cons{F}_{C_G}(v)$ guarantees correct evaluation at gate $v = \land(u,w)$, if $v\in\Cir_v,u\in\Cir_u,w\in\Cir_w$ are all connected to the same $\PPP$-coloured vertex, and thus, in the same copy of $C_G$.
        \item From $\cons{F_{\QQQ}}$ follows that $G$ contains as subgraphs $k$ different $n$-stars with a central $\QQQ$-coloured vertex and leafs coloured by each $\Bit_i$, for $i\in [0,n-1]$. Analogously, we refer to them as $\QQQ$-stars. Because the leafs are connected with the value gadget, we treat them as encodings of a binary string $\bin{x}\in\{0,1\}^n$.
        \item From $\cons{F_{\text{in}}}$ follows that copies of $C_G$ receive at most one bit value per input: for $v\in\{v_0,\dots,v_{2n-1}\}$, satisfying $\cons{F_{\text{in}}}$ requires each input vertex coloured $\Cir_{v}$ to be adjacent to a single vertex coloured $\Bit_{\text{idx}(v)}$ or $\Bit_{\text{idx}(v)-n}$. Because of the constraint $(\smallAlphaCfour{\text{idx}(v)}, 0)$, the values of input vertices have to be identical to the values of their $\Bit_i$-coloured neighbour.
        \item Together with the previous two bullets, satisfying $\cons{F_{\text{str}}}$ suffices to establish that the first (and the second) $n$ input vertices of the same $C_G$ copy must be connected to leafs of the same $\QQQ$-star. Otherwise, we would not reach the counts of $N$ subgraphs $F_{\text{str}_1}$ and $N$ $F_{\text{str}_2}$ with the $N$ different $\PPP$-stars we have available. Due to $\cons{F_{\text{in}}}$, a single $\PPP$-star contributes at most to one each of $F_{\text{str}_1}$ and $F_{\text{str}_2}$.
        \item We observe that $G$ satisfying $(\smallPtwoOut, k(k-1)/2)$ is equivalent to $\val{v_o}=1$ for each vertex $v_o$ corresponding to output node $o$ of a $C_G$ copy, and thus to $C_G(\bin{i},\bin{j})=1$, for their respective inputs $\bin{i},\bin{j}\in\{0,1\}^{n}$.
        \item The previous bullets established that the input of each $C_G$ copy is one of the $k$ inputs encoded by $\QQQ$-stars. Finally, we observe that if each $\QQQ$-star input does not get passed to $k-1$ distinct copies of $C_G$, as either the first or second input, then at least one circuit $C_G(\bin{i},\bin{i})$ receives the same index $\bin{i}$ of vertex $v_i$ as input and output. Since $C_G$ is the succinct representation of a coloured graph without self loops, $C_G(\bin{i},\bin{i}) = 0$. This is in contradiction to $(\smallPtwoOut, k(k-1)/2)$, as less than $k(k-1)/2$ copies of $C_G$ would evaluate to $1$. From the existence of $k$ $n$-bit strings $\bin{i_1},\dots,\bin{i_n}$ that pairwise evaluate to $C_G(\bin{i},\bin{j})=1$, we conclude that the graph encoded by SCR $C_G$ must contain $k$ vertices that are mutually adjacent by $k(k-1)/2)$ edges in total.
    \end{itemize}

    We have computed $O(\abs{C_G})$ constraints that mostly consist of graphs of at most 8 vertices -- with two exceptions of $O(\abs{C_G})$ vertices each -- and used $O(k^2)$ numbers as homomorphism counts. In total, the graphs contain $O(\abs{C_G})$ different colours, which can be encoded in binary strings of length $O(\log\abs{C_G})$ each. This brings the instance produced by our reduction to $O(\abs{C_G}\log\abs{C_G}+\abs{C_G}\log k)$ size. It is not too difficult to compute the constraints in polynomial time from the input circuit $C_G$.
    \end{proof}

    \begin{proof}[Proof of \cref{claim:m:4-correct}.]
        Let $G=(V(G),E({G}),\SSS^{G},\TTT^{G},\XXX^{G},\AAA^{G})$ be some $4$-coloured graph satisfying the constraints.  We will show that $G=(G')^*$ (Remark: rather, $H=(G')^*$, for some subgraph $H\subseteq G$.), i.e that we can undo the previous construction in a way to define an $m$-coloured graph $G'$, which satisfies the constraints $\cons{F}$.
    \begin{itemize}
        \item First, we observe $G$ contains an $m$-colour gadget as induced subgraph, and all other vertices must be $\AAA$-coloured.
        \item Crucially, if $G$ contains an $\AAA$-coloured vertex $v$ that is not connected to any $\XXX$-coloured vertex, then we can simply ignore this vertex for satisfying homomorphism constraints in $(F_i^*,n_i)_{i\in \ell}$ -- any homomorphism from $F_i^*$ to $G$ requires that all $\AAA$-coloured vertices in the image of $F_i^*$ are attached to at least one $\XXX$-coloured vertex.
        \item We define the graph $H$ as the subgraph of $G$ induced by all vertices in $S^{G}$, $T^{G}$, and $C^{G}$ (the colour gadget), as well as any $v\in \AAA^{G}$ with $\abs{N(v)\cap \XXX^G}>0$.
        \item Next, we will define $G'=(V(G'),E(G'),\AAA_1^{G'},\dots,\AAA_m^{G'})$ such that $H=(G')^*$. Let
        \begin{itemize}
            \item $V(G') \coloneqq A^H$,
            \item $E(G') \coloneqq E(H)\cap (\AAA^G)^2$,
            \item $\AAA_i^{G'} \coloneqq \{v \in \AAA^H \mid \text{ $vs_i\in E(H)$} \}$.
        \end{itemize}
        Here, we require that $\abs{N(v)\cap \XXX^G}=1$ for each $v\in V(H)$. This is guaranteed, because $H$ satisfies the constraints $(C_{m,i,j},0)_{i<j\in[m]}$, and any vertex with $\abs{N(v)\cap \XXX^G}>1$ would create a cycle, thus making some of these homomorphism counts from modified $m$-colour gadgets greater than $0$.
        \item It remains to show that $G'$ satisfies any constraint $(F_i,m_i)$ in $\cons{F}$. Because $G$ satisfies the constraint $(F_i^*,m_i)$ by our assumption, and $H$ is the subgraph of $G$ that contains the colour classes $\SSS^{G}$, $\TTT^{G}$, $\XXX^{G}$, and all vertices of $\AAA^G$ that have at least one neighbour in $\XXX^{G}$, any image of $F_i^*$ must be contained in $H$, and thus, $H$ satisfies $(F_i^*,m_i)$ as well. Since $H=(G')^*$, by \cref{claim:4-colour-encoding}, the graph $G'$ satisfies the constraint $(F_i,m_i)$.
    \end{itemize}
    We can construct this list of $4$-coloured constraints in polynomial time, since we require $O(m^2)$ additional constraints, and the construction of $F^*$ from $F$ adds $O(m)$ vertices and $O(m+\abs{V(F)})$ edges to $F$.
    \end{proof}
    
    \begin{proof}[Proof sketch of \cref{lemma:4-to-1}.]
        The proof goes along the lines of uncoloured hardness for $\mathsf{NP}^\mathsf{\#P}$ from \cite{BokerHRSS24}, but differs in details. The main idea is still to introduce a finite set of Kneser graphs as gadgets, which are homomorphism incomparable with each other. Encoding constraint graphs from $F$ with the construction involving these gadgets, essentially, we end up encoding colours as Kneser graphs attached to vertices in the graph $G$.
        
        We make a slight modification to the set $\CI_2$ (from \cite{BokerHRSS24Arxiv}, Appendix C.4) during this proof, because in our current case we can make less assumptions about the underlying graph structure, which was guaranteed to be a $\smallKthree$ in \cite{BokerHRSS24} instead. On the other hand, with only $4$ colours, the number of indicator/(non-)edge/colour$(v)$/colour$(w)$ gadget combinations is no longer dependant on the size of the input. By considering all non-injective images of every pair of these gadgets in our constraints, we can ensure that the graph $G$ must contain the desired gadgets for colours as disjoint subgraphs.
    \end{proof}

\section{Missing Proofs from Section \ref{sec:algorithms}}
\label{sec:appendix_algorithms}

For the proof of Lemma \ref{lem:recursive_counts} we need Pascal's Identity and the Binomial Theorem.
\begin{theorem}[Pascal's Identity {\cite[p.~158]{graham1994concrete}}]\label{thm:pascal}
	For all positive integers $n,k$ it holds that
	\[
	\binom{n}{k}=\binom{n-1}{k}+\binom{n-1}{k-1}.
	\]
\end{theorem}

\begin{theorem}[Binomial Theorem {\cite[p.~162]{graham1994concrete}}]
	For all integers $n \ge 0$ and any real (or complex) numbers $x,y$ it holds that
	\[
	(x+y)^n = \sum_{k=0}^{n} \binom{n}{k} \, x^{\,k} y^{\,n-k}.
	\]
\end{theorem}

\begin{proof}[Proof of Lemma \ref{lem:recursive_counts}]
	We have to show $\s_j(d)=\binom{d}{j}$ for the subgraph version.
	This can be seen via double induction.
In the induction step
	\begin{align}
		\s_j(d)&=\s_j(d-1)+\s_{j-1}(d-1)\label{eq:e1}\\
		&=\binom{d-1}{j}+\binom{d-1}{j-1}\label{eq:e2}\\
		&=\binom{d}{j}\label{eq:e3}
	\end{align}
	we use the definition of $\s_j$ in Equation (\ref{eq:e1}), the induction hypothesis in Equation (\ref{eq:e2}), and Pascal's Identity in Equation (\ref{eq:e3}).
    
The proof for the homomorphism version is done analogously but in the induction step
	\begin{align}
		\h_j(d)&=\sum_{i=0}^{j}\binom{j}{i}\,\h_i(d-1)\\
		&=\sum_{i=0}^{j}\binom{j}{i}\,(d-1)^i\\
		&=\sum_{i=0}^{j}\binom{j}{i}\,(d-1)^i\cdot1^{j-i}\\
		&=(d-1+1)^j\label{eq:e4}\\
		&=d^j
	\end{align}
	we use the Binomial Theorem in Equation (\ref{eq:e4}) instead of Pascal's Identity.
\end{proof}

\begin{proof}[Proof of Lemma \ref{lem:degree_sequence_step}]
Let $(d_1,\dots,d_n)$ be a non-increasing sequence of positive integers with $n\geq2$ and $d_1\leq n-1$.
Let
\[
(d'_1,\dots,d'_{n'})\coloneqq(d_2-1,\dots,d_n-1,\underbrace{1,\dots,1}_{\mathclap{n-1-d_1\text{ times}}}).
\]
We have to show that $(d_1,\dots,d_n)$ is graphic if and only if $(d'_1,\dots,d'_{n'})$, ordered decreasingly, is graphic.

First, observe $\sum_{i=1}^n d_i$ is even if and only if $\sum_{i=1}^{n'} d'_i$ is even since
\begin{align*}
\sum_{i=1}^{n'} d'_i
&=n-1-d_1+\sum_{i=2}^n (d_i-1)\\
&=n-2d_1+\sum_{i=1}^n (d_i-1)\\
&=\underbrace{2(n-d_1)}_{\text{even}}+\sum_{i=1}^n d_i.
\end{align*}

Next, we show the equivalence of
\begin{align}
\sum_{i=1}^k d_i \leq k(k-1) + \sum_{i=k+1}^n \min(d_i,k) \text{ \ \ for $1\leq k\leq n$}
\end{align}
and
\begin{align}
\sum_{i=1}^{k'} d'_i \leq k'(k'-1) + \sum_{i=k'+1}^{n'} \min(d'_i,k') \text{ \ \ for $1\leq k'\leq n'$}.
\end{align}
These inequalities have a direct correspondence with $k'=k-1$ for $k\geq2$ and $k'\leq n-1$ which we show first.
We get for the left-hand side
\begin{align*}
\sum_{i=1}^{k'} d'_i
=\sum_{i=1}^{k-1}d'_i
=\sum_{i=2}^k (d_i-1)
=-k+2-\sum_{i=2}^k d_i
=-k+2-d_1+\sum_{i=1}^k d_i,
\end{align*}
for the first summand of the right-hand side
\begin{align*}
k'(k'-1)
=(k-1)(k-2)
=-2k+2+k(k-1),
\end{align*}
and for the second summand on the right-hand side
\begin{align*}
\sum_{i={k'}+1}^{n'} \min(d'_i,k')
&=\sum_{i=k}^{n'} \min(d'_i,k-1)\\
&=n+1-d_1\sum_{i=k+1}^{n} \min(d_i-1,k-1)\\
&=n-1-d_1 -(n-(k+1))+ \sum_{i=k+1}^{n} \min(d_i,k)\\
&=k-d_1 +\sum_{i=k+1}^{n} \min(d_i,k).
\end{align*}
Using these three equalities we get
\begin{align*}
&&\sum_{i=1}^{k'} d'_i\ \ \ \leq&\ \ \ k'(k'-1) + \sum_{i={k'}+1}^{n'} \min(d'_i,k')\\
\Leftrightarrow&&-k+2-d_1+\sum_{i=1}^k d_i\ \ \ \leq&\ \ \ -2k+2+k(k-1) + k-d_1 +\sum_{i=k+1}^{n} \min(d_i,k)\\
\Leftrightarrow&&\sum_{i=1}^{k} d_i\ \ \ \leq&\ \ \ k(k-1) + \sum_{i={k}+1}^{n} \min(d_i,k).
\end{align*}

Now, for $k=1$ we observe that
\begin{align*}
\underbrace{\sum_{i=1}^{k} d_i}_{\leq d_1} \leq \underbrace{k(k-1)}_{0} + \underbrace{\sum_{i={k}+1}^{n} \min(d_i,k)}_{n-1}
\end{align*}
is always true because we we already assumed $d_1\leq n-1$ as specified by Lemma \ref{lem:degree_sequence_step}.

Finally, for $k'\geq n$ we show that
\begin{align*}
&\sum_{i=1}^{k'} d'_i \leq k'(k'-1) + \sum_{i={k'}+1}^{n'} \min(d'_i,k')\\
\end{align*}
follows from the inequality with $k=n$, which is
\begin{align}
\label{eq:e5}
\sum_{i=1}^{n} d_i \leq n(n-1) +\sum_{i=n+1}^{n} \min(d_i,n).
\end{align}
Indeed, assuming inequality (\ref{eq:e5}) holds we have
\begin{align*}
\sum_{i=1}^{k'} d'_i
=&\sum_{i=1}^{n-1} d'_i + \sum_{i=n}^{k'} d'_i\\\displaybreak[3]
=&\sum_{i=1}^{n-1} d'_i + k'-n+1\\\displaybreak[3]
=&\sum_{i=2}^{n} (d_i-1) + k'-n+1\\\displaybreak[3]
=&\sum_{i=2}^{n} d_i + k'-2n+2\\\displaybreak[3]
=&\sum_{i=1}^{n} d_i + k' - 2n+2-d_1\\\displaybreak[3]
\leq& n(n-1) +\sum_{i=n+1}^{n} \min(d_i,n)+k' - 2n+2-d_1\\\displaybreak[3]
=& n(n-1)+k' - 2n+2-d_1\\\displaybreak[3]
=& n(n-2)+k' - n+2-d_1\\\displaybreak[3]
=& \underbrace{n(n-2)}_{\leq k'(k'-2)} + \underbrace{k'}_{\leq n'}  \underbrace{- n+2}_{\leq 0}\underbrace{-d_1}_{\leq 0}\\\displaybreak[3]
\leq& k'(k'-2) + n'\\\displaybreak[3]
=& k'(k'-1) + n'-k'\\\displaybreak[3]
=& k'(k'-1) + \sum_{i={k'}+1}^{n'} \min(d'_i,k').
\end{align*}
\end{proof}

\begin{proof}[Proof of Lemma \ref{lem:recursive_counts_for_entire_graph_with_slack}]
        Let $(d_1',\dots,d_{n'}')\coloneqq(d_2-1,\dots,d_n-1,\underbrace{1,\dots,1}_{x\text{ times}})$.
        Note that $n'=n-1+x$.
        We have
    \begin{align*}
        \sub(S_0,G)&=n\\
        &=\underbrace{\s_0(d_1)}_{1}+\underbrace{\sub(S_0,G')}_{n-1+x}-x.
    \end{align*}
    And with
    \begin{align*}
        \sum_{i=1}^{n}\s_j(d_i)&=\s_j(d_1)+\sum_{i=2}^{n}\s_j(d_i)\\\displaybreak[2]
        &=\s_j(d_1)+\sum_{i=1}^{n-1}\s_j(d'_i+1)\\\displaybreak[2]
        &=\s_j(d_1)+\sum_{i=1}^{n-1}\left(\s_j(d'_i)+\s_{j-1}(d'_i)\right)\\\displaybreak[2]
        &=\s_j(d_1)+\sum_{i=1}^{n-1}\s_j(d'_i)+\sum_{i=1}^{n-1}\s_{j-1}(d'_i)\ \ \ \text{for $j\geq1$}
    \end{align*}
    we have
    \begin{align*}
        \sub(S_j,G)&=\sum_{i=1}^{n}\s_j(d_i)\\\displaybreak[2]
        &=\s_j(d_1)+\sum_{i=1}^{n-1}\s_j(d'_i)+\sum_{i=n}^{n'}\underbrace{\s_j(1)}_{0}+\sum_{i=1}^{n-1}\s_{j-1}(d'_i)+\sum_{i=n}^{n'}\underbrace{\s_{j-1}(1)}_{0}\\\displaybreak[2]
        &=\s_j(d_1)+\sum_{i=1}^{n'}\s_j(d'_i)+\sum_{i=1}^{n'}\s_{j-1}(d'_i)\\\displaybreak[2]
        &=\s_j(d_1)+\sub\left(S_j,G'\right)+\sub\left(S_{j-1},G'\right)\ \ \ \text{for $j\geq3$},
\end{align*}
    \begin{align*}
        \sub(S_2,G)&=\sum_{i=1}^{n}\s_2(d_i)\\\displaybreak[2]
        &=\s_2(d_1)+\sum_{i=1}^{n-1}\s_2(d'_i)+\sum_{i=1}^{n-1}\s_{1}(d'_i)\\\displaybreak[2]
        &=\s_2(d_1)+\sum_{i=1}^{n-1}\s_2(d'_i)+\underbrace{\sum_{i=n}^{n'}\underbrace{\s_2(1)}_{0}}_{0}+\sum_{i=1}^{n-1}\s_{1}(d'_i)+\underbrace{\sum_{i=n}^{n'}\underbrace{\s_{1}(1)}_{1}}_{x}-x\\\displaybreak[2]
        &=\s_2(d_1)+\sum_{i=1}^{n'}\s_2(d'_i)+\sum_{i=1}^{n'}\s_{1}(d'_i)-x\\\displaybreak[2]
        &=\s_2(d_1)+\sub\left(S_2,G'\right)+2\cdot\sub\left(S_1,G'\right)-x.
    \end{align*}
    Analogously, with the factor 2 coming from Observation \ref{obs:graph_counts_to_vertex_counts}, we have
    \begin{align*}
        2\cdot\sub(S_1,G)&=\sum_{i=1}^{n}\s_1(d_i)\\\displaybreak[2]
        &=\s_1(d_1)+\sum_{i=1}^{n-1}\s_1(d'_i)+\sum_{i=1}^{n-1}\s_{0}(d'_i)\\\displaybreak[2]
        &=\s_1(d_1)+\sum_{i=1}^{n-1}\s_1(d'_i)+\underbrace{\sum_{i=n}^{n'}\underbrace{\s_1(1)}_{1}}_{x}+\sum_{i=1}^{n-1}\s_{0}(d'_i)+\underbrace{\sum_{i=n}^{n'}\underbrace{\s_{0}(1)}_{1}}_{x}-2x\\\displaybreak[2]
        &=\s_1(d_1)+\sum_{i=1}^{n'}\s_1(d'_i)+\sum_{i=1}^{n'}\s_{0}(d'_i)-2x\\\displaybreak[2]
        &=\s_1(d_1)+2\cdot\sub\left(S_1,G'\right)+\sub\left(S_0,G'\right)-2x.
    \end{align*}
\end{proof}

To prove Lemma \ref{lem:problem_recursion} we will combine Lemma \ref{lem:degree_sequence_step} with Lemma \ref{lem:recursive_counts_for_entire_graph_with_slack} as follows.
\begin{lemma}
	\label{lem:recursive_equivalence_for_entire_graph_with_slack}
	Let $(d_1,\dots,d_n)$ be a non-increasing sequence of positive integers with $n\geq2$.
	There is a graph $G$ with the degree sequence $(d_1,\dots,d_n)$ if and only if there is a graph $G'$ with the degree sequence
	\[
	(d_2-1,\dots,d_n-1,\underbrace{1,\dots,1}_{\mathclap{x\coloneqq n-1-d_1\text{ times}}}),
	\]
	ordered decreasingly, and
	\begin{align*}
		\sub(S_0,G)
		=&\s_0(d_1)+\sub\left(S_0,G'\right)-x,\\
		2\cdot\sub(S_1,G)
		=&\s_1(d_1)+2\cdot\sub\left(S_1,G'\right)+\sub\left(S_{0},G'\right)-2x,\\
		\sub(S_2,G)
		=&\s_2(d_1)+\sub\left(S_2,G'\right)+2\cdot\sub\left(S_{1},G'\right)-x,\\
		\sub(S_j,G)
		=&\s_j(d_1)+\sub\left(S_j,G'\right)+\sub\left(S_{j-1},G'\right)\ \ \ \ \ \text{for $j\geq3$}.
	\end{align*}
    \begin{proof}
        This follows immediately from Lemma \ref{lem:degree_sequence_step} and Lemma \ref{lem:recursive_counts_for_entire_graph_with_slack}.
    \end{proof}
\end{lemma}

\begin{proof}[Proof of Lemma \ref{lem:problem_recursion}]
	First, we assume $\dpf(s_0,\dots,s_\ell,m,x)=1$, so by Definition \ref{def:recursive_solve} there exists integers $d_1\ge\dots\ge d_{s_0}\geq0$ and a graph $G$ with degree sequence
	\[
	(d_1,\dots,d_n)=(d_1,\dots,d_{s_0},\underbrace{1,\dots,1}_{x\text{ times}}),
	\]
	ordered decreasingly, which fulfils
	\begin{align*}
		&\sub(S_0,G)
		=\s_0+x,\\
		2\cdot&\sub(S_1,G)
		=\s_1+x,\\
		&\sub(S_j,G)
		=\s_j\ \ \ \ \ \text{for $j\in[2,\ell]$},\\
		&d_i\leq m\ \ \ \ \ \text{for $i\in[1,s_0]$}.
	\end{align*}
	Let $(d'_1,\dots,d'_{n'})$ be the sequence $(d_1,\dots,d_n)$ after removing all entries of value $0$.
	Obviously, there is a graph $G'$  with the degree sequence $(d'_1,\dots,d'_{n'})$ and
	\begin{align*}
		&\sub(S_0,G')
		=\s_0+x-n+n',\\
		2\cdot&\sub(S_1,G')
		=\s_1+x,\\
		&\sub(S_j,G')
		=\s_j\ \ \ \ \ \text{for $j\in[2,\ell]$},\\
		&d'_i\leq d'_1\ \ \ \ \ \text{for $i\in[1,n']$}.
	\end{align*}
	We also know $n'\geq2$ because otherwise $x=0$ but then $(d_1,...,d_n)=(d_1,0,\dots,0)$ with $d_1\geq1$ (due to $s_1\geq1$) would not be graphic which is a contradiction.
	So, according to Lemma \ref{lem:recursive_equivalence_for_entire_graph_with_slack} there is a graph $G''$ with degree sequence
	\begin{align*}
    (d'_2-1,\dots,d'_{n'}-1,\underbrace{1,\dots,1}_{\mathclap{n'-1-d'_1\text{ times}}}),
	\end{align*}
	ordered decreasingly, for which we have
	\begin{align*}
		&\sub(S_0,G'')
		=\s_0-s_0(d'_1)-n+n'+x+(n'-1-d'_1),\\
		2\cdot&\sub(S_1,G'')
		=\s_1-s_1(d'_1)-\sub(S_0,G'')+x+2(n'-1-d'_1),\\
		&\sub(S_2,G'')
		=\s_2-s_2(d'_1)-2\cdot\sub(S_1,G'')+(n'-1-d'_1),\\
		&\sub(S_j,G'')
		=\s_j-s_j(d'_1)-\sub(S_{j-1},G'')\ \ \ \ \ \text{for $j\in[3,\ell]$}.
	\end{align*}
Note that $G''$ contains at least $x$ vertices of degree 0 which we will prune next.
    Formally, let
	\begin{align*}
    (d''_1,\dots,d''_{n''})\coloneqq(d'_2-1,\dots,d'_{n'-x}-1,\underbrace{1,\dots,1}_{\mathclap{n'-1-d'_1\text{ times}}}).
	\end{align*}
    Obviously, there is a graph $G'''$ with the degree sequence $(d''_1,\dots,d''_{n''})$, ordered decreasingly, for which we have
	\begin{align*}
		&\sub(S_0,G''')
		=\s_0-s_0(d'_1)-n+n'+(n'-1-d'_1),\\
		2\cdot&\sub(S_1,G''')
		=\s_1-s_1(d'_1)-\sub(S_0,G''')+2(n'-1-d'_1),\\
		&\sub(S_2,G''')
		=\s_2-s_2(d'_1)-2\cdot\sub(S_1,G''')+(n'-1-d'_1),\\
		&\sub(S_j,G''')
		=\s_j-s_j(d'_1)-\sub(S_{j-1},G''')\ \ \ \ \ \text{for $j\in[3,\ell]$},\\
		&d''_i\leq d'_1-1\ \ \ \ \ \text{for $i\in[1,n'']$}.
	\end{align*}
	Now, for 
	\begin{align*}
		&s'_0
		\coloneqq s_0-s_0(d'_1)-n+n',\\
		&s'_k
		\coloneqq s_k-s_k(d'_1)-s'_{k-1}\ \ \ \ \ \text{for $k\in[1,\ell]$},\\
		&m'\coloneqq d'_1-1,\\
		&x'\coloneqq n'-1-d'_1
	\end{align*}
	we have $\dpf(s'_0,\dots,s'_\ell,m',x')=1$ because there exists the graph $G'''$ with degree sequence
	\[
	(d''_1,\dots,d''_{n''})=(d''_1,\dots,d''_{s'_0},\underbrace{1,\dots,1}_{x'\text{ times}}),
	\]
	ordered decreasingly, which fulfils
	\begin{align*}
		&\sub(S_0,G''')
		=\s'_0+x',\\
		2\cdot&\sub(S_1,G''')
		=\s'_1+x',\\
		&\sub(S_j,G''')
		=\s'_j\ \ \ \ \ \text{for $j\in[2,\ell]$},\\
		&d''_i\leq m'\ \ \ \ \ \text{for $i\in[1,n'']$},
	\end{align*}
	and with $d=d'_1$ and $n_1=s_0-n+n'$ we have
	\begin{align*}
		&d\in[1,m],\\\displaybreak[2]
		&n_1\in [1,s_0],\\\displaybreak[2]
		&d\leq n_1-1+x,\\\displaybreak[2]
		&s'_0=n_1-1,\\\displaybreak[2]
		&s'_j=s_j-s_j(d)-s'_{j-1} \text{ for $j\in[1,\ell]$},\\\displaybreak[2]
		&m'=d-1,\\\displaybreak[2]
		&x'=x+n_1-1-d,
	\end{align*}
	which is exactly what we had to show for the first direction of the proof.
	
	Since Lemma \ref{lem:recursive_equivalence_for_entire_graph_with_slack} is an equivalence as are the pruning steps in this proof, the other direction follows similarly.
\end{proof}

\end{document}